\documentclass[lettersize,journal]{IEEEtran}
\usepackage{booktabs,siunitx}
\usepackage{amsthm}
\usepackage{amssymb}
\usepackage{amsmath}
\usepackage{flushend}
\usepackage{array}
\usepackage{multirow}
\usepackage{graphicx}
\usepackage{wrapfig}
\usepackage{balance}
\usepackage{algorithm}
\usepackage{algpseudocode}
\usepackage{threeparttable}
\usepackage{comment}
\newtheorem{theorem}{Theorem}[section]
\newtheorem{note}[theorem]{\textbf{Note.}}
\newtheorem{definition}[theorem]{\textbf{Def.}}
\newtheorem{proposition}[theorem]{\textbf{Prop.}}
\usepackage{todonotes}
\usepackage{tikz}
\usetikzlibrary{matrix}
\usetikzlibrary{calc}
\usetikzlibrary{shapes.geometric}
\tikzstyle{every node}=[font=\footnotesize]
\tikzstyle{etichetta}=[]
\tikzstyle{square}=[draw,outer sep=0pt,inner sep=-1.3em,regular polygon,regular polygon sides=4,,minimum size=13.5mm]
\tikzstyle{square2}=[square,dotted,black!50]
\tikzstyle{square3}=[draw,outer sep=0pt,inner sep=-1.3em,regular polygon,regular polygon sides=4,minimum size=8mm]
\tikzstyle{square4}=[square3,dotted,black!50]
\pgfmathsetmacro{\ra}{0.95}
\begin{document}
\title{MementoHash: A Stateful, Minimal Memory, Best Performing Consistent Hash Algorithm}

\author{
    \IEEEauthorblockN{
        Massimo Coluzzi,
        Amos Brocco,
        Alessandro Antonucci,
        Tiziano Leidi}\\
    \IEEEauthorblockA{
        Department of Innovative Technologies,\\
        University of Applied Sciences and Arts of Southern Switzerland,
        Lugano, Switzerland\\
    Email: \{massimo.coluzzi,amos.brocco,alessandro.antonucci,tiziano.leidi\}@supsi.ch}
}

\maketitle
\begin{abstract} Consistent hashing is used in distributed systems and networking applications to spread data evenly and efficiently across a cluster of nodes. In this paper, we present \emph{MementoHash}, a novel consistent hashing algorithm that eliminates known limitations of state-of-the-art algorithms while keeping optimal performance and minimal memory usage. We describe the algorithm in detail, provide a pseudo-code implementation, and formally establish its solid theoretical guarantees. To measure the efficacy of \emph{MementoHash}, we compare its performance, in terms of memory usage and lookup time, to that of state-of-the-art algorithms, namely, \emph{AnchorHash}, \emph{DxHash}, and \emph{JumpHash}. Unlike \emph{JumpHash}, \emph{MementoHash} can handle random failures. Moreover, \emph{MementoHash} does not require fixing the overall capacity of the cluster (as \emph{AnchorHash} and \emph{DxHash} do), allowing it to scale indefinitely. The number of removed nodes affects the performance of all the considered algorithms. Therefore, we conduct experiments considering three different scenarios: stable (no removed nodes), one-shot removals (90\% of the nodes removed at once), and incremental removals. We report experimental results that averaged a varying number of nodes from ten to one million. Results indicate that our algorithm shows optimal lookup performance and minimal memory usage in its best-case scenario. It behaves better than \emph{AnchorHash} and \emph{DxHash} in its average-case scenario and at least as well as those two algorithms in its worst-case scenario. 
\end{abstract}

\begin{IEEEkeywords}
Consistent hashing, load balancing, scalability.
\end{IEEEkeywords}

\section{Introduction}
A distributed system consists of multiple nodes that manage different kind of data, such as files for distributed storage, records for distributed databases, or requests for load balancers. An even distribution of these data units among the nodes is necessary to prevent specific nodes from being overloaded. Consistent hashing is widely known for its ability to evenly allocate the load across the system and minimize the number of data units that need to be remapped during cluster scaling \cite{saxenaanalysis}.
After the introduction of cloud infrastructures, \emph{elasticity}, namely the ability to scale quickly and efficiently, has become a concept of paramount importance. Elasticity demands efficient and high-performing consistent hashing algorithms because data units must be redistributed to maintain balance as nodes are added or removed.
Each data unit is uniquely identified by a key, and each node is mapped to a sequential integer called a \textit{bucket}. In the considered scenario, each node in a cluster of size $n$ is mapped to a unique value in the range $[0,n\!-\!1]$. The lookup operation maps each key to a bucket in a deterministic way, therefore invoking the lookup operation on the same key must return the same bucket as long as such a bucket is available.
The challenge is to efficiently map keys to buckets, knowing that each key represents a data unit and each bucket represents a node in a distributed system.
The scientific literature proposes several algorithms that strive to obtain the best balance between desirable properties, algorithmic complexity, and real-world performance. 
Our contribution, named \emph{MementoHash}, represents a novel consistent hashing algorithm that eliminates known limitations of state-of-the-art algorithms while keeping optimal performance and minimal memory usage.

\section{Related work}
\textit{Consistent hashing} algorithms are not a novel concept: the first example in the literature can be dated back to $1996$, when Thaler and Ravishankar proposed \textit{Rendezvous} \cite{thaler1996rendezvous}\cite{thaler1998using}, while the term was first used by Karger and co-authors in $1997$ \cite{karger1997consistent}. Other consistent hashing algorithms for non peer-to-peer environments followed. Overall, the most prominent ones, published between $1996$ and $2021$, are:

\begin{enumerate}
    \item[-] \textit{Rendezvous}: published by Thaler and Ravishankar in 1996 \cite{thaler1996rendezvous}\cite{thaler1998using}. 
    \item[-] \textit{Consistent Hashing Ring}: published by Karger et al. in 1997 \cite{karger1997consistent}\cite{karger1999consistent}.   
    \item[-] \textit{JumpHash}: published by Lamping and Veach in 2014 \cite{lamping2014fast}.
    \item[-] \textit{Multi-probe}: published by Appleton and O’Reilly in 2015 \cite{appleton2015multi}.
    \item[-] \textit{Maglev}: published by Eisenbud in 2016 \cite{eisenbud2016maglev}.
    \item[-] \textit{AnchorHash}: published by Mendelson et al. in 2020 \cite{mendelson2020anchorhash}.
    \item[-] \textit{DxHash}: published by Dong and Wang in 2021 \cite{dong2021dxhash}.
\end{enumerate}

For the sake of conciseness, in the following we omit the \textit{hash} suffix from the names of the aforementioned algorithms (e.g., \textit{Memento} instead of \textit{MementoHash}).

In a previous work \cite{posterIspass2023}\cite{sebd2023}, we implemented all the above algorithms in Java together with a benchmark tool specifically designed for consistent hashing algorithms \cite{isinGitHub}.
We compared them against the following metrics:
    \begin{itemize}
        \item[-] \textbf{memory usage}: the amount of memory used to store the internal data structure.
        \item[-] \textbf{initialization time}: the time needed to initialize the internal data structure.
        \item[-] \textbf{lookup time}: the time needed to map a given key to its related bucket.
        \item[-] \textbf{resize time}: the time to change the internal data structure when adding or removing buckets.
        \item[-] \textbf{balance}: the ability to spread the keys evenly among the buckets.
        \item[-] \textbf{resize balance}: the ability to spread the keys evenly among the buckets after resizing the cluster.
        \item[-] \textbf{monotonicity}: the ability to move only the keys involved in the resizing.
    \end{itemize}

We found that \textit{Anchor}, \textit{Dx}, and \textit{Jump} outmatch the other algorithms in all the considered metrics.
In particular, \textit{Jump} was shown to be the best-performing algorithm since it does not use any internal data structure. It uses minimal memory and is the fastest in any time metric because it does not access memory and runs at CPU speed. 
\textit{Jump} allows only the last inserted bucket to be removed, meaning that it is not able to handle the failure of a random node in the cluster. In practical settings, any computational resource, whether hardware or software-based, has the potential to experience failures. Even in failure-transparent architectures, where every node is guaranteed to be highly available, network failures can make one or more nodes unreachable. Therefore, despite its excellent performance, \textit{Jump} is an impractical solution for production environments.
\textit{Anchor} and \textit{Dx} address this limitation by using an internal data structure to keep track of all the cluster nodes (both working and not working) which causes them to use much more memory than \textit{Jump} and to be slower in all the time metrics. Moreover, those two algorithms require the overall capacity of the cluster to be defined during initialization. The capacity cannot change during execution, forcing an upper bound to the scalability of the cluster.
We argue that it is unnecessary to maintain a record of all potential nodes, even those that are not operational. While conducting our comparative analysis, we pondered whether it would be feasible to harness the outstanding capabilities of \emph{Jump} by incorporating a minimal data structure that only remembers the nodes that have failed. This line of thinking led us to the creation of \emph{MementoHash}, which overcomes the limitations of \emph{Jump} while maintaining nearly identical performance.

\section{Preliminaries}\label{preliminaries}
As previously mentioned, we consider the problem of distributing data units, which are uniquely identified by keys, to buckets. In this context we commonly turn our attention to hashing algorithms, which are deterministic functions that take an arbitrary amount of data as input and produce a fixed length output called \textit{hash value} or \textit{digest}. The digest, usually a number, can be mapped to a value in the interval $[0,n\!-\!1]$ using modular algebra. Hashing algorithms serve the purpose of evenly distributing keys among buckets, ensuring a balanced load distribution across all buckets. Additionally, they facilitate efficient determination of the mapping between keys and buckets. However, when the number of buckets is altered, a simple approach based on hash functions and modular algebra will result in a significant remapping of keys to different buckets. Consequently, a distributed system experiences the relocation of nearly all its data units across nodes each time a node joins or departs from the cluster. To address this problem, consistent hashing solutions have been devised. It is a class of distributed hashing algorithms that provide the following properties:

\begin{itemize}
    \item[] \textbf{balance}: keys are evenly distributed among buckets. Given $k$ keys and $n$ buckets, ideally, $\frac{k}{n}$ keys are mapped to each bucket.
    \item[] \textbf{minimal disruption}: the same key is always mapped to the same bucket as long as such a bucket is available. When a bucket leaves the cluster, only the keys mapped to such a bucket will move to other buckets, while keys previously mapped to the other buckets will not move.
    \item[] \textbf{monotonicity}: when a new bucket is added, keys only move from an existing bucket to the new one, but not from an existing one to another. Given $k$ keys and $n$ buckets, ideally, only $\frac{k}{n+1}$ keys should move to the new bucket.
\end{itemize}

Since many consistent hashing algorithms rely on non-consistent hashing functions we deem important to introduce the following assumption.

\begin{note}\label{note:uniformHashFunctions}
\textbf{(Uniform hash functions)} We assume the hash functions used inside the consistent hashing algorithms to produce a uniform distribution of the keys.
\end{note}

\subsection{Notation}
As stated in previous sections, we can map each node of a distributed system to an integer in the range $[0,n\!-\!1]$ called a \emph{bucket}. Therefore, we can represent a cluster as an array of buckets.

\begin{definition}[\textbf{b-array}]
With the term \emph{b-array} we refer to the array of buckets representing a cluster of nodes in a distributed system. A \emph{b-array} $\mathcal{N}$ of size $n$ is an array of integers where every position contains a value in the range $[0,n\!-\!1]$
\end{definition}

\begin{definition}[\textbf{working bucket}]
With the term \textit{working bucket} we identify a bucket related to a working node.
\end{definition}

Moving forward, this paper will adhere to the following notation:
\begin{definition}[\textbf{notation}]
We will use $\subset$ to denote a strict inclusion and $\subseteq$ to denote a non-strict inclusion.\\
Given a b-array:
\begin{itemize}
    \item[-] $\mathcal{N}:=[0,\dots,n\!-\!1]$ is the content of the b-array and $n := |\mathcal{N}|$ is its size.
    \item[-] $\mathcal{W}\subseteq\mathcal{N}$ is the set of working buckets.
    \item[-] $\mathcal{W}_b:=\mathcal{W}\backslash\{b\}$ is the set of working buckets after removing bucket $b$. 
    \item[-] $\langle b\rightarrow c,p \rangle$ represents the replacement of the failing bucket $b$ with the bucket $c$ (more details in Sec.~\ref{sec:MementoHash}).
    \item[-] $\mathcal{R} := \{ \langle b \rightarrow c, p \rangle | b,c\!\in\!\mathcal{N}, p\!\leq\!n\}$ is the set of replacements and $r := |\mathcal{R}|$ is the number of replacements.
\end{itemize}
\end{definition}

\section{State of the art}
\label{sec: state-of-the-art}
In this section, we will briefly introduce the \textit{Jump}, \textit{Anchor}, and \textit{Dx} algorithms. We will discuss their strengths and weaknesses to allow a better understanding of the improvements introduced by \textit{Memento}.
 
\subsection{The Jump algorithm}
\label{sec:jump}

By representing a cluster as a b-array, \textit{Jump} \cite{lamping2014fast} assumes every bucket to be working and the b-array to be sorted.

\begin{figure}[H]
\centering
\begin{tikzpicture}[]
\node [square] (0) at (1,2) {$0$};
\node [square] (1) at (1+\ra,2) {$1$};
\node [square] (2) at (1+2*\ra,2) {$2$};
\node [square] (3) at (1+3*\ra,2) {$3$};
\node [square] (4) at (1+4*\ra,2) {$\ldots$};
\node [square] (5) at (1+5*\ra,2) {$n\!-\!2$};
\node [square] (6) at (1+6*\ra,2) {$n\!-\!1$};
\end{tikzpicture}
\caption{\textit{Jump}'s representation of a cluster}
\label{fig: init b-array}
\end{figure}
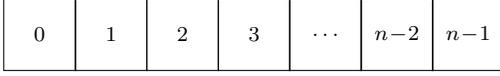

\textit{Jump} assumes buckets are sorted from $0$ to $n\!-\!1$ and cannot change their position. Therefore, it only needs to store the size of the b-array (see, e.g., Fig.~\ref{fig: init b-array}). When the cluster scales up, \textit{Jump} increases the buckets count, assuming new buckets to be added to the tail of the b-array with values $n$, $n\!+\!1$, $\ldots$ (see, e.g., Fig.~\ref{fig: init b-array adding buckets}).
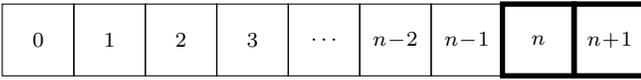
\begin{figure}[H]
\centering
\begin{tikzpicture}[]
\node [square] (0) at (1,2) {$0$};
\node [square] (1) at (1+\ra,2) {$1$};
\node [square] (2) at (1+2*\ra,2) {$2$};
\node [square] (3) at (1+3*\ra,2) {$3$};
\node [square] (4) at (1+4*\ra,2) {$\ldots$};
\node [square] (5) at (1+5*\ra,2) {$n\!-\!2$};
\node [square] (6) at (1+6*\ra,2) {$n\!-\!1$};
\node [square,line width=2]  (7) at (1+7*\ra,2) {$n$};
\node [square,line width=2]  (8) at (1+8*\ra,2) {$n\!+\!1$};
\end{tikzpicture}
\caption{Jump Hash: Adding two buckets}
\label{fig: init b-array adding buckets}
\end{figure}

When the cluster scales down, \textit{Jump} reduces the buckets count, assuming buckets to be removed from the tail, first $n\!-\!1$, then $n\!-\!2$, ... (see, e.g., Fig.~\ref{fig: init b-array removing buckets})

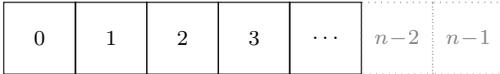
\begin{figure}[H]
\centering
\begin{tikzpicture}[]
\node [square] (0) at (1,2) {$0$};
\node [square] (1) at (1+\ra,2) {$1$};
\node [square] (2) at (1+2*\ra,2) {$2$};
\node [square] (3) at (1+3*\ra,2) {$3$};
\node [square] (4) at (1+4*\ra,2) {$\ldots$};
\node [square2] (5) at (1+5*\ra,2) {\color{black!50}{$n\!-\!2$}};
\node [square2] (6) at (1+6*\ra,2) {\color{black!50}{$n\!-\!1$}};
\end{tikzpicture}
\caption{Jump Hash: Removing two buckets}
\label{fig: init b-array removing buckets}
\end{figure}

The algorithm takes a key and the size of the b-array as input parameters and returns the bucket such a key belongs (e.g., $jump(key,n) \rightarrow b\in\mathcal{N}$).
Its output depends only on the provided parameters. \textit{Jump} holds the properties of \textit{balance}, \textit{monotonicity}, and \textit{minimal disruption}. As an example, let us assume $jump(key,10) \rightarrow 5$, then also $jump(key,9) \rightarrow 5$, $jump(key,8) \rightarrow 5$, $\ldots$, $jump(key,6) \rightarrow 5$.
The bucket returned by \textit{Jump} will change only when $5$ is no more a working bucket (e.g., $jump(key,5) \rightarrow 2$ ).
Each key is always mapped to the same bucket until such a bucket is removed. This behavior describes the property of \textit{minimal disruption}.
\textit{Jump} operates without an internal state and does not perform any memory access, allowing it to run at CPU speed. In spite of its slower theoretical growth rate, this characteristic makes \textit{Jump} the fastest algorithm in practical terms \cite{posterIspass2023}\cite{sebd2023}.
In an ideal scenario where failures are absent, it would be possible to scale our cluster by adding and removing buckets in a \emph{Last-In-First-Out} (LIFO) order. In such a case, \textit{Jump} would be the most effective consistent-hashing algorithm available. However, in real-world environments, any node can fail at any time. Furthermore, we could need to shut down for maintenance the node mapped to the bucket $0$. \textit{Jump} does not allow you to do that unless you shut down the entire cluster. This limitation makes \textit{Jump} impractical for real-world environments.

\subsection{The Anchor algorithm}
To address the limitations of \textit{Jump}, \textit{Anchor} \cite{mendelson2020anchorhash} represents all the possible buckets up to the maximum size reachable by the cluster. This choice allows \textit{Anchor} to mark which buckets are working and which are not, allowing it to handle random failures (see, e.g., Fig.~\ref{fig: anchor b-array}). Let us call $W_b$, the set of working buckets remaining after removing $b$. When a bucket fails, the algorithm maps $\mathcal{W}_b$ to $b$. Every removed bucket $b$ is mapped to the set of buckets that were working after the removal of $b$.
The basic idea behind \textit{Anchor} is to use a hash function to map a key to one of the buckets $b$ in the cluster. If $b$ is working, the algorithm is done. Otherwise, the key is mapped to a bucket $c$ in the subset of buckets $\mathcal{W}_b$ that were working when $b$ was removed. The algorithm iterates the same logic on $c$. If $c$ is working, the algorithm is done. Otherwise, the key is mapped to a bucket $d$ in the subset of buckets $\mathcal{W}_c$ that worked when $c$ was removed. Since $c\in\mathcal{W}_b$, $c$ worked when $b$ was removed. Therefore, calling $\mathcal{W}$ the set of working buckets, we have that $\mathcal{W}\subset\mathcal{W}_c\subset\mathcal{W}_b$. Every iteration will work on a smaller set, and a working bucket will eventually be found.

\begin{figure}[H]
\centering
\begin{tikzpicture}[scale=.6]
\tikzstyle{every node}=[font=\tiny]
\node [square4]  (1) at (1,2) {$0$};
\node [square3]  (2) at (1+\ra,2) {$1$};
\node [square4]  (3) at (1+2*\ra,2) {$2$};
\node [square3]  (4) at (1+3*\ra,2) {$3$};
\node [square3]  (5) at (1+4*\ra,2) {$4$};
\node [square4]  (6) at (1+5*\ra,2) {$5$};
\node [square3]  (7) at (1+6*\ra,2) {$\ldots$};
\node [square3]  (8) at (1+7*\ra,2) {$n\!-\!2$};
\node [square3]  (9) at (1+8*\ra,2) {$n\!-\!1$};
\node [square4] (10) at (1+9*\ra,2) {$n$};
\node [square4] (11) at (1+10*\ra,2) {$n\!+\!1$};
\node [square4] (12) at (1+11*\ra,2) {$\ldots$};
\node [square4] (13) at (1+12*\ra,2) {$a\!-\!2$};
\node [square4] (14) at (1+13*\ra,2) {$a\!-\!1$};
\end{tikzpicture}
\caption{Cluster representation for Anchor}
\label{fig: anchor b-array}
\end{figure}
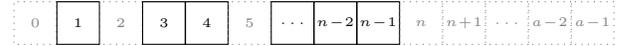

The described solution takes $O(ln(\frac{a}{w}))$ iterations to look up a key, where $a$ is the overall capacity of the cluster (working and nonworking buckets). Still, it uses a considerable amount of memory. To overcome this issue, the authors of \textit{Anchor} suggest an in-place solution that leverages four arrays of integers to store the required information. The in-place solution uses less memory but takes $O([ln(\frac{a}{w})]^2)$ steps to perform a lookup. Both versions rely on the assumption of keeping track of all the available nodes in the cluster.
The size of the internal data structure cannot change to ensure \textit{monotonicity} and \textit{minimal disruption}. Therefore, from start the arrays must be big enough to contain any possible buckets.
Foreseeing from the early beginning the maximum size reachable by the cluster is not a trivial task. Capacity can be ensured by instantiating large arrays, but overestimating will consume resources and slow down performance.

\subsection{The Dx algorithm}
\textit{Dx} \cite{dong2021dxhash} reduces the memory consumption by using a bit-array to mark the availability of the buckets. It is a remarkable improvement compared to the four integer arrays used by \textit{Anchor}. Still, it uses an amount of memory proportional to the overall capacity of the cluster. In order to compute a position in the range $[0,a\!-\!1]$, \textit{Dx} uses a pseudo-random function $R$ initialized with the key as the seed. Accordingly, a sequence of buckets in the form $R(k), R(R(k)), R(R(R(k)))...$ can be generated, and the first working bucket is chosen (see, e.g., Fig.~\ref{fig: dx lookup}).
\textit{Dx} trades the improvement in memory consumption in terms of lookup time. It takes $O(\frac{a}{w})$ steps to perform a lookup while suffering from the same limitations of \textit{Anchor} (i.e., upper bound in the overall capacity of the cluster).

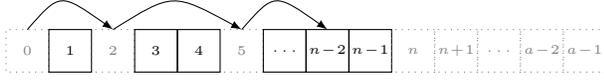
\begin{figure}[H]
\centering
\begin{tikzpicture}[scale=.6]
\tikzstyle{every node}=[font=\tiny]
\node [square4]   (1) at (1,2) {$0$};
\node [square3]   (2) at (1+\ra,2) {$1$};
\node [square4]   (3) at (1+2*\ra,2) {$2$};
\node [square3]   (4) at (1+3*\ra,2) {$3$};
\node [square3]   (5) at (1+4*\ra,2) {$4$};
\node [square4]   (6) at (1+5*\ra,2) {$5$};
\node [square3]   (7) at (1+6*\ra,2) {$\ldots$};
\node [square3]   (8) at (1+7*\ra,2) {$n\!-\!2$};
\node [square3]   (9) at (1+8*\ra,2) {$n\!-\!1$};
\node [square4]  (10) at (1+9*\ra,2) {$n$};
\node [square4]  (11) at (1+10*\ra,2) {$n\!+\!1$};
\node [square4]  (12) at (1+11*\ra,2) {$\ldots$};
\node [square4]  (13) at (1+12*\ra,2) {$a\!-\!2$};
\node [square4]  (14) at (1+13*\ra,2) {$a\!-\!1$};
\draw [black,->,-latex] (1.north) .. controls (0.3+1.5*\ra,3.3) .. (3.north);
\draw [black,->,-latex] (3.north) .. controls (0.9+3.5*\ra,3.3) .. (6.north);
\draw [black,->,-latex] (6.north) .. controls (0.3+6.5*\ra,3.3) .. (8.north);
\end{tikzpicture}
\caption{Lookup process of Dx}
\label{fig: dx lookup}
\end{figure}

\section{MementoHash}\label{sec:MementoHash}
The basic idea behind the \textit{Memento} algorithm is to use memory just to remember the removed buckets. It starts with an empty data structure and relies on \textit{Jump} as its core engine. When all the buckets are working, or when buckets are removed in \textit{LIFO} order, \textit{Memento} works exactly like \textit{Jump}. If a random bucket is removed (i.e., the related node fails) \textit{Memento} stores this information in its internal data structure and redistributes the keys among the remaining buckets.

\subsection{Dense b-arrays}
When we initially set up a new cluster, we associate each node to a sequential bucket. We can represent the cluster as a b-array where each bucket represents a working node and is positioned at the corresponding index. \textit{Jump} assumes the b-array always to be in this configuration (see, e.g., Fig.~\ref{fig: N0}). 

\begin{figure}[H]
\centering
\begin{tikzpicture}[scale=.6]
\tikzstyle{every node}=[font=\tiny]
\node [etichetta] (0) at (-.5,2.8) {Index};
\node [etichetta] (0) at (-.5,2) {Bucket};
\node [etichetta] (0) at (1,2.8) {$0$};
\node [etichetta] (1) at (1+\ra,2.8) {$1$};
\node [etichetta] (2) at (1+2*\ra,2.8) {$2$};
\node [etichetta] (3) at (1+3*\ra,2.8) {$3$};
\node [etichetta] (4) at (1+4*\ra,2.8) {$4$};
\node [etichetta] (5) at (1+5*\ra,2.8) {$5$};
\node [etichetta] (6) at (1+6*\ra,2.8) {$6$};
\node [etichetta] (7) at (1+7*\ra,2.8) {$7$};
\node [etichetta] (8) at (1+8*\ra,2.8) {$8$};
\node [etichetta] (9) at (1+9*\ra,2.8) {$9$};
\node [square3] (0) at (1,2) {$0$};
\node [square3] (1) at (1+\ra,2) {$1$};
\node [square3] (2) at (1+2*\ra,2) {$2$};
\node [square3] (3) at (1+3*\ra,2) {$3$};
\node [square3] (4) at (1+4*\ra,2) {$4$};
\node [square3] (5) at (1+5*\ra,2) {$5$};
\node [square3] (6) at (1+6*\ra,2) {$6$};
\node [square3] (7) at (1+7*\ra,2) {$7$};
\node [square3] (8) at (1+8*\ra,2) {$8$};
\node [square3] (9) at (1+9*\ra,2) {$9$};
\end{tikzpicture}
\caption{Initial state of a b-array of size 10}
\label{fig: N0}
\end{figure}
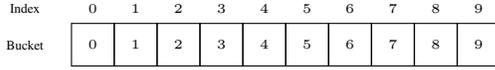

\begin{note}\label{node:accessByIndex}
For clarity in the upcoming explanations, remember that \textit{a b-array} $\mathcal{N}$ is always accessed by index (or position), just like any regular array. When we use the term \textit{bucket} we refer to the identifier of the corresponding resource, while when using \textit{position} or \textit{index} we refer to a location within the \textit{b-array}.
In a \textit{b-array} of size $n$, each bucket $b$ has a value within the range $[0, n\!-\!1]$. As a result, $b$ can also serve as an index for accessing the value $\mathcal{N}[b]$. Initially, each bucket is located at its corresponding index, i.e., $\mathcal{N}[b] = b\  \forall b$ (refer to Fig.~\ref{fig: N0}). When a bucket $b$ is removed, it gets replaced by a bucket $c$. Therefore, $\mathcal{N}[b] = c$.
\end{note}

\begin{definition}[\textbf{dense b-array}]
Given a b-array $\mathcal{N}$ of size $n$ and a value $w \leq n$. We define the b-array as \emph{dense} up to $w$, $D_w(\mathcal{N})$ if every index between $0$ and $w\!-\!1$ contains a working bucket.
\end{definition}

The b-array representing the initial configuration of a cluster of size $n=w=|\mathcal{N}|$ is dense (i.e., $D_n(\mathcal{N})$ holds), and the buckets are sorted from $0$ to $n\!-\!1$. If we remove the last bucket, $w=|\mathcal{N}|-1$, hence $D_{n-1}(\mathcal{N})$ holds. It is the working condition for \textit{Jump}. On the other hand, if we remove a random bucket other than the last (i.e., the related node fails), we create a gap, and neither $D_n(\mathcal{N})$ nor $D_{n-1}(\mathcal{N})$ hold.

\subsection{Maintaining a dense b-array}
In the general sense, random node removals prevent us from using \textit{Jump}.
In the following we thus present the basic idea that enables us to keep the b-array dense.
Let us assume to have a b-array $\mathcal{N}_0$ of size $10$ (Fig.~\ref{fig: N0-N1-N2}), $D_{10}(\mathcal{N}_0)$ holds. If we remove bucket $9$, the condition $D_{9}(\mathcal{N}_1)$ holds. Subsequently, if we remove bucket $5$, we create a gap, and neither $D_{9}(\mathcal{N}_1)$ nor $D_8(\mathcal{N}_1)$ would hold. Accordingly, to fill the gap, we can copy bucket $8$ in position $5$ to obtain $\mathcal{N}_2$, where $D_8(\mathcal{N}_2)$ and $D_{9}(\mathcal{N}_2)$ hold again.

\begin{figure}[H]
\centering
\begin{tikzpicture}[scale=.6]
\tikzstyle{every node}=[font=\tiny]
\node [etichetta] (0) at (-.1,2) {\normalsize $\mathcal{N}_0$};
\node [etichetta] (0) at (-.1,0) {\normalsize $\mathcal{N}_1$};
\node [etichetta] (0) at (-.1,-2) {\normalsize $\mathcal{N}_2$};
\node [etichetta,align=left,text width = 8mm] (0) at (0.4+11*\ra,2) {Initial state};
\node [etichetta,align=left,text width = 8mm] (0) at (0.4+11*\ra,0) {Removal of node 9};
\node [etichetta,align=left,text width = 8mm] (0) at (0.4+11*\ra,-2) {Replace bucket 5 with 8};

\node [etichetta] (0) at (1,2.8) {$0$};
\node [etichetta] (1) at (1+\ra,2.8) {$1$};
\node [etichetta] (2) at (1+2*\ra,2.8) {$2$};
\node [etichetta] (3) at (1+3*\ra,2.8) {$3$};
\node [etichetta] (4) at (1+4*\ra,2.8) {$4$};
\node [etichetta] (5) at (1+5*\ra,2.8) {$5$};
\node [etichetta] (6) at (1+6*\ra,2.8) {$6$};
\node [etichetta] (7) at (1+7*\ra,2.8) {$7$};
\node [etichetta] (8) at (1+8*\ra,2.8) {$8$};
\node [etichetta] (9) at (1+9*\ra,2.8) {$9$};
\node [square3] (0) at (1,2) {$0$};
\node [square3] (1) at (1+\ra,2) {$1$};
\node [square3] (2) at (1+2*\ra,2) {$2$};
\node [square3] (3) at (1+3*\ra,2) {$3$};
\node [square3] (4) at (1+4*\ra,2) {$4$};
\node [square3] (5) at (1+5*\ra,2) {$5$};
\node [square3] (6) at (1+6*\ra,2) {$6$};
\node [square3] (7) at (1+7*\ra,2) {$7$};
\node [square3] (8) at (1+8*\ra,2) {$8$};
\node [square3] (9) at (1+9*\ra,2) {$9$};
\draw [|-|] (1-.5*\ra,1.2) -- (1+9.5*\ra,1.2);
\node [etichetta]  (0) at (4.5*\ra,1.0) {$n=10$, $w=10$};

\node [square3] (0) at (1,0) {$0$};
\node [square3] (1) at (1+\ra,0) {$1$};
\node [square3] (2) at (1+2*\ra,0) {$2$};
\node [square3] (3) at (1+3*\ra,0) {$3$};
\node [square3] (4) at (1+4*\ra,0) {$4$};
\node [square3] (5) at (1+5*\ra,0) {$5$};
\node [square3] (6) at (1+6*\ra,0) {$6$};
\node [square3] (7) at (1+7*\ra,0) {$7$};
\node [square3] (8) at (1+8*\ra,0) {$8$};
\node [square4] (9) at (1+9*\ra,0) {$9$};
\draw [|-|] (1-.5*\ra,-.8) -- (1+8.5*\ra,-.8);
\node [etichetta]  (0) at (4.5*\ra,-1) {$n=9$, $w=9$};
\node [square3] (0) at (1,-2) {$0$};
\node [square3] (1) at (1+\ra,-2) {$1$};
\node [square3] (2) at (1+2*\ra,-2) {$2$};
\node [square3] (3) at (1+3*\ra,-2) {$3$};
\node [square3] (4) at (1+4*\ra,-2) {$4$};
\node [square3] (5c) at (1+5*\ra,-2) {$\mathbf{8}$};
\node [square3] (6) at (1+6*\ra,-2) {$6$};
\node [square3] (7) at (1+7*\ra,-2) {$7$};
\node [square3] (8c) at (1+8*\ra,-2) {$8$};
\node [square4] (9) at (1+9*\ra,-2) {$9$};
\draw [black,->,-latex] (8c.north) .. controls (1+6.5*\ra+1.4,-0.9) and (1+6.5*\ra-1.4,-0.9) .. (5c.north);
\draw [|-|] (1-.5*\ra,-2.8) -- (1+7.5*\ra,-2.8);
\node [etichetta]  (0) at (4.5*\ra,-3) {$n=9$, $w=8$};
\end{tikzpicture}
\caption{Removing a bucket: the last (9) or any other than the last (5)}
\label{fig: N0-N1-N2}
\end{figure}
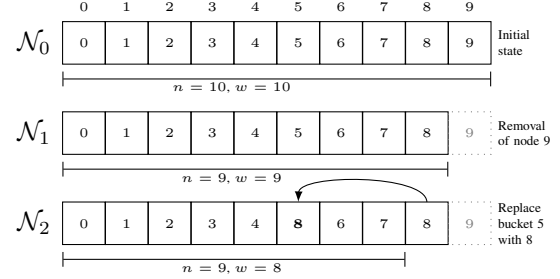

Since we removed bucket $5$, the new size of our cluster is $w=8$. So, we will consider only the buckets in the positions from $0$ to $7$. After removing bucket $5$ the working buckets are $[0,8] \backslash \{5\}$, which are precisely the buckets available from position $0$ to $7$. If we remove another bucket, for example, bucket $1$, we fill the gap by copying bucket $7$ in position $1$ to obtain $\mathcal{N}_3$ (Fig. \ref{fig: N3}). The working buckets are now $[0,8] \backslash \{1,5\}$, and the size of the cluster becomes $7$.

\begin{proposition} \label{prop:lastRepacingIsNewWorkingSize}
The new size of the cluster always corresponds to the index of the bucket we use to fill the gap.
\end{proposition}

\begin{figure}[H]
\centering
\begin{tikzpicture}[scale=.6]
\tikzstyle{every node}=[font=\tiny]
\node [etichetta]  (0) at (-.1,-2) {\normalsize $\mathcal{N}_3$};
\node [etichetta,align=left,text width = 8mm]  (0) at (0.4+11*\ra,-2) {Replace bucket 1 with 7};
\node [square3] (0) at (1,-2) {$0$};
\node [square3] (1c) at (1+\ra,-2) {$\mathbf{7}$};
\node [square3] (2) at (1+2*\ra,-2) {$2$};
\node [square3] (3) at (1+3*\ra,-2) {$3$};
\node [square3] (4) at (1+4*\ra,-2) {$4$};
\node [square3] (5c) at (1+5*\ra,-2) {$\mathbf{8}$};
\node [square3] (6) at (1+6*\ra,-2) {$6$};
\node [square3] (7c) at (1+7*\ra,-2) {$7$};
\node [square3] (8c) at (1+8*\ra,-2) {$8$};
\node [square4] (9) at (1+9*\ra,-2) {$9$};
\draw [black,->,-latex] (7c.north) .. controls (0.2+6.5*\ra,-0.5) and (0.2+3.5*\ra,-0.5) .. (1c.north);
\draw [|-|] (1-.5*\ra,-2.8) -- (1+6.5*\ra,-2.8);
\node [etichetta]  (0) at (4.5*\ra,-3) {$n=9$, $w=7$};
\end{tikzpicture}
\caption{Filling the gap after removal of bucket 1}
\label{fig: N3}
\end{figure}
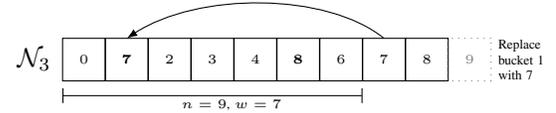

If we compare the b-array versions $\mathcal{N}_1$, $\mathcal{N}_2$, and $\mathcal{N}_3$ we can notice that they are identical except for the positions of the removed buckets. Therefore, instead of representing the whole b-array, we can represent only the differences between the initial and the current state of the b-array (Fig. \ref{fig: N1-N2-N3 Comparison}).

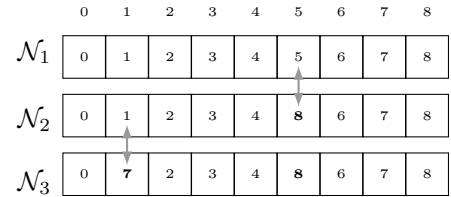
\begin{figure}[H]
\centering
\begin{tikzpicture}[scale=.6]
\tikzstyle{every node}=[font=\tiny]
\node [etichetta] (0) at (-.1,2.1) {\normalsize $\mathcal{N}_1$};
\node [etichetta] (0) at (-.1,0.6) {\normalsize $\mathcal{N}_2$};
\node [etichetta] (0) at (-.1,-.8) {\normalsize $\mathcal{N}_3$};

\node [etichetta] (0) at (1,3) {$0$};
\node [etichetta] (1) at (1+\ra,3) {$1$};
\node [etichetta] (2) at (1+2*\ra,3) {$2$};
\node [etichetta] (3) at (1+3*\ra,3) {$3$};
\node [etichetta] (4) at (1+4*\ra,3) {$4$};
\node [etichetta] (5) at (1+5*\ra,3) {$5$};
\node [etichetta] (6) at (1+6*\ra,3) {$6$};
\node [etichetta] (7) at (1+7*\ra,3) {$7$};
\node [etichetta] (8) at (1+8*\ra,3) {$8$};
\node [square3] (0) at (1,2) {$0$};
\node [square3] (1) at (1+\ra,2) {$1$};
\node [square3] (2) at (1+2*\ra,2) {$2$};
\node [square3] (3) at (1+3*\ra,2) {$3$};
\node [square3] (4) at (1+4*\ra,2) {$4$};
\node [square3] (5) at (1+5*\ra,2) {$5$};
\node [square3] (6) at (1+6*\ra,2) {$6$};
\node [square3] (7) at (1+7*\ra,2) {$7$};
\node [square3] (8) at (1+8*\ra,2) {$8$};
\node [square3] (0) at (1,0.7) {$0$};
\node [square3] (1) at (1+\ra,0.7) {$1$};
\node [square3] (2) at (1+2*\ra,0.7) {$2$};
\node [square3] (3) at (1+3*\ra,0.7) {$3$};
\node [square3] (4) at (1+4*\ra,0.7) {$4$};
\node [square3] (5) at (1+5*\ra,0.7) {$\mathbf{8}$};
\node [square3] (6) at (1+6*\ra,0.7) {$6$};
\node [square3] (7) at (1+7*\ra,0.7) {$7$};
\node [square3] (8) at (1+8*\ra,0.7) {$8$};
\node [square3] (0) at (1,-.6) {$0$};
\node [square3] (1) at (1+\ra,-.6) {$\mathbf{7}$};
\node [square3] (2) at (1+2*\ra,-.6) {$2$};
\node [square3] (3) at (1+3*\ra,-.6) {$3$};
\node [square3] (4) at (1+4*\ra,-.6) {$4$};
\node [square3] (5) at (1+5*\ra,-.6) {$\mathbf{8}$};
\node [square3] (6) at (1+6*\ra,-.6) {$6$};
\node [square3] (7) at (1+7*\ra,-.6) {$7$};
\node [square3] (8) at (1+8*\ra,-.6) {$8$};
\draw [black,thick,<->,latex-latex,black!40] (1+\ra,0.5) -- (1+\ra,-0.4);
\draw [black,thick,<->,latex-latex,black!40] (1+5*\ra,1.8) -- (1+5*\ra,0.9);
\end{tikzpicture}
\caption{Comparing versions $\mathcal{N}_1$, $\mathcal{N}_2$, and $\mathcal{N}_3$}
\label{fig: N1-N2-N3 Comparison}
\end{figure}

Each bucket corresponds to its position in the initial state $\mathcal{N}_0$, so we do not need to represent it; we just need the size of the b-array. Removing the last bucket just reduces the size of the b-array. Also in this case, we just need to update its size ($n\gets n\!-\!1$). For every removed bucket $b$, other than the last one, we remember its replacement in the form $b\rightarrow c$ where $c$ is the replacing bucket.
As we will discuss later, when we restore the buckets, we should do it in reverse order, from the last removed bucket backward. This approach preserves the properties of \textit{balance} and \textit{minimal disruption} of the algorithm.
Therefore, when a bucket $b$ is removed, we store the replacement $b \rightarrow c$ together with the previously removed bucket $p$.

\begin{definition}[\textbf{replacement}]\label{def:replacement}
A replacement is defined by a tuple $\langle b \rightarrow c, p\rangle$ where $b$ is the removed bucket, $c$ is the replacing one, and $p$ is the previously removed bucket. 
\end{definition}

\begin{definition}[\textbf{replacement set}]\label{def:replacementset}
The replacement set $\mathcal{R}$ is the data structure used to store the replacements $\langle b\rightarrow c,p\rangle$ for all the removed buckets. We are only interested in knowing if there exists a replacement for a given bucket $b$. Therefore, we can implement $\mathcal{R}$ using a hash table that takes $b$ as the key and returns $\langle c,p\rangle$ as the value.
This allows us to add a replacement for bucket $b$ into $\mathcal{R}$, remove the replacement for bucket $b$ from $\mathcal{R}$, and check if a replacement for the bucket $b$ exists in constant time.
\end{definition}

Given a cluster with $n$ initial nodes, all the buckets are working and are in the range $[0,n\!-\!1]$. When the first bucket $b$ is removed, we set $p = n$ in the replacement. This way, when the first removed bucket gets restored, the next node added to the cluster will be mapped to the bucket $n$ as expected.

\begin{proposition}
The number of working buckets $w=n-r$ is given by the difference between the size of the b-array $n = |\mathcal{N}|$ and the number of replacements $r = |\mathcal{R}|$ (i.e., the number of removed buckets).
\end{proposition}

Let $l$, with $0 \le l \le n$, be the \textit{last removed bucket}.
Accordingly, the aforementioned removals lead to the following updates:
\begin{itemize}
    \item[] \textit{Initialization:} $n=10$, $R_{\mathcal{N}_0} = \{\}$, $l=10$
    \item[] \textit{Removing bucket 9:} $n=9$, $R_{\mathcal{N}_1} = \{\}$, $l=9$
    \item[] \textit{Removing bucket 5}: $n=9$, $R_{\mathcal{N}_2} = \{ \langle 5 \rightarrow 8, 9 \rangle$, $l=5$
    \item[] \textit{Removing bucket 1}: $n=9$, $R_{\mathcal{N}_3} = \{ \langle 5 \rightarrow 8, 9 \rangle, \langle 1 \rightarrow 7, 5 \rangle \}$, $l=1$
\end{itemize}

This replacement strategy works well in all situations. However, in the following sections, we will discuss two edge cases we can get while removing buckets.

\subsection{Removing a replacing bucket}\label{sec:removingAReplacingBucket}
In the previous example, we removed bucket $5$, replacing it with $8$, and then we removed bucket $1$, replacing it with $7$. We will now show what happens if we remove bucket $8$ from $\mathcal{N}_3$.
Following the same logic, we will copy bucket $6$ in position $8$.
The remaining working buckets are $6$ namely: $\mathcal{N}_4 = \{0,2,3,4,6,7\}$ (Fig. \ref{fig:removingAReplacingBucket}).

\begin{figure}[H]
\centering
\begin{tikzpicture}[scale=.6]
\tikzstyle{every node}=[font=\tiny]
\node [etichetta] (0) at (-.1,2) {\normalsize $\mathcal{N}_4$};
\node [etichetta] (0) at (1,2.8) {$0$};
\node [etichetta] (1) at (1+\ra,2.8) {$1$};
\node [etichetta] (2) at (1+2*\ra,2.8) {$2$};
\node [etichetta] (3) at (1+3*\ra,2.8) {$3$};
\node [etichetta] (4) at (1+4*\ra,2.8) {$4$};
\node [etichetta] (5) at (1+5*\ra,2.8) {$5$};
\node [etichetta] (6) at (1+6*\ra,2.8) {$6$};
\node [etichetta] (7) at (1+7*\ra,2.8) {$7$};
\node [etichetta] (8) at (1+8*\ra,2.8) {$8$};
\node [etichetta] (9) at (1+9*\ra,2.8) {$9$};
\node [square3] (0) at (1,2) {$0$};
\node [square3] (1) at (1+\ra,2) {$7$};
\node [square3] (2) at (1+2*\ra,2) {$2$};
\node [square3] (3) at (1+3*\ra,2) {$3$};
\node [square3] (4) at (1+4*\ra,2) {$4$};
\node [square3] (5) at (1+5*\ra,2) {$8$};
\node [square3] (6a) at (1+6*\ra,2) {$6$};
\node [square3] (7) at (1+7*\ra,2) {$7$};
\node [square3] (8a) at (1+8*\ra,2) {$\mathbf{6}$};
\node [square4] (9) at (1+9*\ra,2) {$9$};
\draw [|-|] (1-.5*\ra,1.2) -- (1+5.5*\ra,1.2);
\node [etichetta]  (0) at (3.5*\ra,0.9) {$n=9$, $w=6$};
\draw [black,->,thick,-latex] (6a.north) .. controls (0.2+7.4*\ra,3.3) and (0.2+7.8*\ra,3.3) .. (8a.north);

\end{tikzpicture}
\caption{Removing a replacing bucket}
\label{fig:removingAReplacingBucket}
\end{figure}
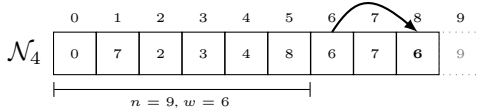

At this point, to find the actual replacement of bucket $5$, we need to apply the substitution two times $5 \rightarrow 8 \rightarrow 6$.
In Fig. \ref{fig:removingAReplacingBucket:lookupSteps}, we can see how substitutions can be chained to obtain a proper replacement.



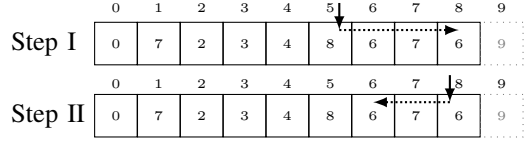
\begin{figure}[H]
\centering
\begin{tikzpicture}[scale=.6]
\tikzstyle{every node}=[font=\tiny]

\node [etichetta]  (0) at (-.6,1.6) {\normalsize Step I};
\node [etichetta]  (0) at (1,2.4) {$0$};
\node [etichetta]  (1) at (1+\ra,2.4) {$1$};
\node [etichetta]  (2) at (1+2*\ra,2.4) {$2$};
\node [etichetta]  (3) at (1+3*\ra,2.4) {$3$};
\node [etichetta]  (4) at (1+4*\ra,2.4) {$4$};
\node [etichetta]  (5) at (1+5*\ra,2.4) {$5$};
\node [etichetta]  (6) at (1+6*\ra,2.4) {$6$};
\node [etichetta]  (7) at (1+7*\ra,2.4) {$7$};
\node [etichetta]  (8) at (1+8*\ra,2.4) {$8$};
\node [etichetta]  (9) at (1+9*\ra,2.4) {$9$};
\node [square3]    (0) at (1,1.6) {$0$};
\node [square3]  (1) at (1+\ra,1.6) {$7$};
\node [square3]  (2) at (1+2*\ra,1.6) {$2$};
\node [square3]  (3) at (1+3*\ra,1.6) {$3$};
\node [square3]  (4) at (1+4*\ra,1.6) {$4$};
\node [square3]  (5) at (1+5*\ra,1.6) {$8$};
\node [square3]  (6a) at (1+6*\ra,1.6) {$6$};
\node [square3]  (7) at (1+7*\ra,1.6) {$7$};
\node [square3]  (8a) at (1+8*\ra,1.6) {$6$};
\node [square4]  (9) at (1+9*\ra,1.6) {$9$};
\draw [black,->,thick,-latex] (1+5*\ra+0.2,2.5) -- (1+5*\ra+0.2,1.9);
\draw [black,densely dotted,->,thick,-latex] (1+5*\ra+0.2,1.9) -- (1+8*\ra,1.9);

\node [etichetta] (0) at (-.5,0) {\normalsize Step II};
\node [etichetta] (0) at (1,0.7) {$0$};
\node [etichetta] (1) at (1+\ra,0.7) {$1$};
\node [etichetta] (2) at (1+2*\ra,0.7) {$2$};
\node [etichetta] (3) at (1+3*\ra,0.7) {$3$};
\node [etichetta] (4) at (1+4*\ra,0.7) {$4$};
\node [etichetta] (5) at (1+5*\ra,0.7) {$5$};
\node [etichetta] (6) at (1+6*\ra,0.7) {$6$};
\node [etichetta] (7) at (1+7*\ra,0.7) {$7$};
\node [etichetta] (8) at (1+8*\ra,0.7) {$8$};
\node [etichetta] (9) at (1+9*\ra,0.7) {$9$};
\node [square3] (0) at (1,0) {$0$};
\node [square3] (1) at (1+\ra,0) {$7$};
\node [square3] (2) at (1+2*\ra,0) {$2$};
\node [square3] (3) at (1+3*\ra,0) {$3$};
\node [square3] (4) at (1+4*\ra,0) {$4$};
\node [square3] (5) at (1+5*\ra,0) {$8$};
\node [square3] (6) at (1+6*\ra,0) {$6$};
\node [square3] (7) at (1+7*\ra,0) {$7$};
\node [square3] (8) at (1+8*\ra,0) {$6$};
\node [square4] (9) at (1+9*\ra,0) {$9$};
\draw [black,->,thick,-latex] (1+8*\ra-0.2,0.9) -- (1+8*\ra-0.2,0.3);
\draw [black,densely dotted,->,thick,-latex] (1+8*\ra-0.2,0.3) -- (1+6*\ra,0.3);

\end{tikzpicture}
\caption{Removing a replacing bucket - lookup steps}
\label{fig:removingAReplacingBucket:lookupSteps}
\end{figure}

The proper bucket positions (containing the working buckets) are in the range $[0,5]$. Therefore, positions $6$, $7$, or $8$ cannot be directly reached. The only way to hit one of the related buckets is through redirection.

The same concept holds in the opposite situation where the replacing bucket was removed in a previous iteration. In this case we fill the gap with an already replaced bucket. Therefore, to find the working bucket we need to follow the chain of replacements.

\subsection{Replacing a bucket with itself}
Suppose we remove bucket $5$ from the b-array $\mathcal{N}_4$. Applying the same logic, bucket $5$ must be copied in position $5$ to obtain $\mathcal{N}_5$ (Fig. \ref{fig:replacingABucketWithItself}). Bucket $5$ is replaced by itself. However, this is not an issue since, after removing bucket $5$, the size of the cluster becomes $5$. Therefore, the proper bucket positions are in the range $[0,4]$, representing the set of working buckets $\{0,2,3,4,7\}$. Position $5$ cannot be hit directly, and bucket $5$ cannot be reached via redirection.

\begin{figure}[H]
\centering
\begin{tikzpicture}[scale=.6]
\tikzstyle{every node}=[font=\tiny]
\node [etichetta] (0) at (-.1,2) {\normalsize $\mathcal{N}_5$};
\node [etichetta] (0) at (1,2.8) {$0$};
\node [etichetta] (1) at (1+\ra,2.8) {$1$};
\node [etichetta] (2) at (1+2*\ra,2.8) {$2$};
\node [etichetta] (3) at (1+3*\ra,2.8) {$3$};
\node [etichetta] (4) at (1+4*\ra,2.8) {$4$};
\node [etichetta] (5) at (1+5*\ra,2.8) {$5$};
\node [etichetta] (6) at (1+6*\ra,2.8) {$6$};
\node [etichetta] (7) at (1+7*\ra,2.8) {$7$};
\node [etichetta] (8d) at (1+8*\ra,2.8) {$8$};
\node [etichetta] (9) at (1+9*\ra,2.8) {$9$};
\node [square3] (0) at (1,2) {$0$};
\node [square3] (1) at (1+\ra,2) {$7$};
\node [square3] (2) at (1+2*\ra,2) {$2$};
\node [square3] (3) at (1+3*\ra,2) {$3$};
\node [square3] (4) at (1+4*\ra,2) {$4$};
\node [square3] (5b) at (1+5*\ra,2) {$\mathbf{5}$};
\node [square3] (6a) at (1+6*\ra,2) {$6$};
\node [square3] (7) at (1+7*\ra,2) {$7$};
\node [square3] (8a) at (1+8*\ra,2) {$6$};
\node [square4] (9) at (1+9*\ra,2) {$9$};
\draw [|-|] (1-.5*\ra,1.2) -- (1+4.5*\ra,1.2);
\node [etichetta]  (0) at (3.5*\ra,0.9) {$n=9$, $w=5$};
\draw [black,-latex,thick] (1.2+4.5*\ra,2.5) .. controls (1+5*\ra-0.3,3.4) and (1+5*\ra+0.3,3.4) .. (0.8+5.5*\ra,2.5);
\end{tikzpicture}
\caption{Replacing a bucket with itself}
\label{fig:replacingABucketWithItself}
\end{figure}
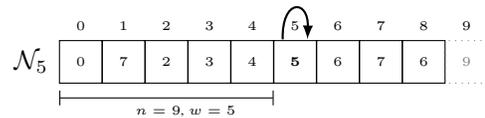

\section{Explanation of the algorithm}
As described in the previous sections, \textit{Memento} relies entirely on \textit{Jump} as its core engine. The latter ensures \textit{balance}, \textit{monotonicity} when adding new buckets, and \textit{minimal disruption} when removing buckets from the tail. However, since \textit{Memento} allows for removing arbitrary buckets, additional state information needs to be stored.

\begin{definition}[\textbf{state}]
The algorithm state information $S$, which comprises all relevant information required by \textit{Memento}, is defined as $\mathcal{S}$ := $\langle n,\mathcal{R},l\rangle$ where:
\begin{itemize}
    \item[] $n$ is the size of the b-array
    \item[] $\mathcal{R}$ is the set of replacements
    \item[] $l$, with $0 \le l \le n$, is the last removed bucket.
\end{itemize}

\end{definition}

In the following, the underlying details of the algorithm will be discussed.

\subsection{Initialization}
When we set up a new cluster of size $n$, all the initial buckets are working ($\mathcal{W} = \mathcal{N}$ and $\mathcal{R} = \emptyset$). \textit{Memento} initializes its state by storing the initial number of buckets and creating an empty set of replacements (Alg.~\ref{algo:init}). The last removed bucket $l$ gets initialized to the value of $n$, since the b-array would grow starting from this index.

\begin{algorithm}[H]
  \caption{Init Memento}
  \begin{algorithmic}
    \Require{$initial\_node\_count > 0$}
    \Function{INIT}{$initial\_node\_count$}
      \State $n \gets initial\_node\_count$
      \State $l \gets n$ \Comment{l = last removed bucket}
      \State $\mathcal{R} \gets \emptyset$
    \EndFunction
  \end{algorithmic}
  \label{algo:init}
\end{algorithm}

\subsection{Removing a bucket}
If all the buckets work and we remove a bucket from the tail, it reduces the b-array's size. Otherwise, if there are already removed buckets or when removing a bucket $b$ other than the last, \textit{Memento} creates a new replacement $\langle b \rightarrow w\!-\!1, l \rangle$ and puts it into $\mathcal{R}$. The bucket $b$ becomes the new last removed $l \leftarrow b$. The removal procedure is described in Alg.~\ref{algo:remove}.

\begin{algorithm}[ht]
  \caption{Remove a bucket}
  \label{algo:remove}
  \begin{algorithmic}
    \Require{$0 \leq b < n$}
    \Function{REM}{$b$}
      \If{$b = n-1$ \textbf{and} $\mathcal{R} = \emptyset$}
        \State $n \gets n-1$ \Comment{Size is reduced by 1}
      \Else
        \State $w \gets n - |\mathcal{R}|$
        \State $\mathcal{R} \gets \mathcal{R} \cup \langle b \rightarrow w\!-\!1, l \rangle$
      \EndIf
      \State $l \gets b$ \Comment{The last removed bucket}
    \EndFunction
  \end{algorithmic}
\end{algorithm}

\subsection{Adding a bucket}
If all the buckets are working, a new bucket will be added to the tail of the array, increasing the number of buckets by one. Otherwise, we will restore the last removed bucket (Alg.~\ref{algo:add}).

\begin{algorithm}[ht]
  \caption{Add a bucket}
  \label{algo:add}
  \begin{algorithmic}
    \Ensure{$0 \leq b \leq n$}
    \Function{ADD()}{}
    \If{$\mathcal{R} = \emptyset$} \Comment{Adds a bucket to the tail}
        \State $b \gets n$
        \State $l \gets n \gets n+1$ 
      \Else
        \State $b \gets l$ \Comment{Get the last removed bucket}
        \State $\mathcal{R} \gets \mathcal{R} \backslash \langle b \rightarrow c,p \rangle$ \Comment{Remove the replacement}
        \State $l \gets p$ \Comment{Updates the last removed bucket}
      \EndIf
      \State \Return $b$
    \EndFunction
  \end{algorithmic}
\end{algorithm}

As described in Sec.~\ref{sec:removingAReplacingBucket}, there are edge cases where removing buckets may create chains of replacements. By always restoring the last removed bucket, we guarantee to properly untie the chains of replacements.

\subsection{Lookup}\label{sec:lookup}
The explanation of the lookup function is the following. We start hitting a position in the b-array. If the related bucket $b$ works, we are done. Otherwise, we hit a new position in the portion of b-array $\mathcal{W}_b$ containing the working buckets when $b$ was removed. We hit a new bucket, $c\in\mathcal{W}_b$. If it works, we are done. Otherwise, we search in the portion of b-array $\mathcal{W}_c$ containing the working buckets when $c$ was removed. This second portion is smaller than $\mathcal{W}_b$ (i.e., $\mathcal{W}\subseteq\mathcal{W}_c\subset\mathcal{W}_b$). Therefore, eventually, we will reach the portion of the b-array $\mathcal{W}$ containing only working buckets.

More specifically, when performing a lookup, \textit{Memento} starts by invoking \textit{Jump} to get a bucket $b$ in the range $[0,n\!-\!1]$ (Alg.~\ref{alg:lookup} line 2). If the bucket is working, the lookup is done. Otherwise, there must be a replacement in the form $\langle b \rightarrow c,p \rangle \in \mathcal{R}$ for some $c,p \in \mathcal{N}$.

When we start a new cluster, all the buckets are working. \textit{Jump} is balanced and evenly maps the keys among the buckets. When we remove a bucket $b$ other than the last, all the keys originally mapped to $b$ need to be evenly distributed among the remaining buckets.

As stated in Prop.~\ref{prop:lastRepacingIsNewWorkingSize}, the value $c$ corresponds also to the remaining number of working buckets after the removal of $b$ (Alg.~\ref{alg:lookup} line 4). Therefore, the working buckets can be found only in range $[0,c\!-\!1]$.

We hash the key again to get a new mapping in the range $[0,c\!-\!1]$ (Alg.~\ref{alg:lookup} lines 5 and 6). This time we do not need the hashing algorithm to be consistent. So, we can reduce complexity by using an uniform hash function as described in Note~\ref{note:uniformHashFunctions} that works in $O(1)$ (assuming keys to be of fixed size).

\begin{algorithm}
  \caption{Lookup}
  \label{alg:lookup}
  \begin{algorithmic}[1]
    \Ensure{$0 \leq b < n$}
    \Function{LOOKUP}{$key$}
        \State $b \gets jump(key,n)$
        \While{$\exists \langle b \rightarrow c,p \rangle \in \mathcal{R}$}
            \State $w_b \gets c$ \Comment{Working buckets after b was removed}
            \State $h \gets hash(key,b)$ \Comment{Traditional hash function}
            \State $d \gets h$ mod $w_b$ \Comment{New bucket in $[0,w_b\!-\!1]$}
            \While{$\exists \langle d \rightarrow u,q \rangle \in \mathcal{R}$ and $u \geq w_b$} 
                \State $d \gets u$ \Comment{Follow substitutions}
            \EndWhile
            \State $b \gets d$
        \EndWhile
      \State \Return $b$
    \EndFunction
  \end{algorithmic}
\end{algorithm}

At this step, if we hit a working bucket, the algorithm is done. Otherwise, we follow the replacements chain until either we find a working bucket, or we find a bucket $d$ removed after $b$ (Alg.~\ref{alg:lookup} lines 7 and 8). It is worth noting that the condition $u \geq w_b$ of the inner loop contributes in guaranteeing the balance. Since, without such a condition, the inner loop would always follow the replacement chain until the end resulting in an unbalanced distribution.
Suppose, for instance, to have a b-array of size 6 and remove buckets 0, 3, and 5 in order (Fig.~\ref{fig:remove 0,3,5 in order}, the three steps denote the removal order and the corresponding replacement). The replacement set is $\mathcal{R} = \{\langle 0 \rightarrow 5,6 \rangle, \langle 3 \rightarrow 4,0 \rangle\, \langle 5 \rightarrow 3, 3 \rangle\}$.

\begin{figure}[H]
\centering
\begin{tikzpicture}[scale=1]
\tikzstyle{every node}=[font=\tiny]

\node [etichetta]  (1l) at (-.5,3.2) {\normalsize $Step\ I$};
\node [etichetta]  (1l) at (-.5,2.7) {\normalsize $\langle 0\rightarrow 5,6\rangle$};
\node [etichetta]  (10) at (1-0.2,3.7) {$0$};
\node [etichetta]  (11) at (1+\ra,3.7) {$1$};
\node [etichetta]  (12) at (1+2*\ra,3.7) {$2$};
\node [etichetta]  (13) at (1+3*\ra,3.7) {$3$};
\node [etichetta]  (14) at (1+4*\ra,3.7) {$4$};
\node [etichetta]  (15) at (1+5*\ra+0.2,3.7) {$5$};
\node [square]  (10) at (1,3) {$\mathbf{5}$};
\node [square]  (11) at (1+\ra,3) {$1$};
\node [square]  (12) at (1+2*\ra,3) {$2$};
\node [square]  (13) at (1+3*\ra,3) {$4$};
\node [square]  (14) at (1+4*\ra,3) {$4$};
\node [square]  (15) at (1+5*\ra,3) {$5$};
\draw [black,->,thick,-latex] (15.north) .. controls (1+5*\ra,4.3) and (1+0*\ra+0.2,4.3) .. (10.north);

\node [etichetta]  (2l) at (-.5,1.7) {\normalsize $Step\ II$};
\node [etichetta]  (2l) at (-.5,1.2) {\normalsize $\langle 3\rightarrow 4,0\rangle$};
\node [etichetta]  (20) at (1,2.2) {$0$};
\node [etichetta]  (21) at (1+\ra,2.2) {$1$};
\node [etichetta]  (22) at (1+2*\ra,2.2) {$2$};
\node [etichetta]  (23) at (1+3*\ra-0.2,2.2) {$3$};
\node [etichetta]  (24) at (1+4*\ra+0.2,2.2) {$4$};
\node [etichetta]  (25) at (1+5*\ra,2.2) {$5$};
\node [square]  (20) at (1,1.5) {$5$};
\node [square]  (21) at (1+\ra,1.5) {$1$};
\node [square]  (22) at (1+2*\ra,1.5) {$2$};
\node [square]  (23) at (1+3*\ra,1.5) {$\mathbf{4}$};
\node [square]  (24) at (1+4*\ra,1.5) {$4$};
\node [square]  (25) at (1+5*\ra,1.5) {$5$};
\draw [black,->,thick,-latex] (24.north) .. controls (1+4*\ra,2.5) and (1+3*\ra+0.2,2.5) .. (23.north);

\node [etichetta]  (3l) at (-.5,0.2) {\normalsize $Step\ III$};
\node [etichetta]  (3l) at (-.5,-0.3) {\normalsize $\langle 5\rightarrow 3,3\rangle$};
\node [etichetta]  (30) at (1,0.7) {$0$};
\node [etichetta]  (31) at (1+\ra,0.7) {$1$};
\node [etichetta]  (32) at (1+2*\ra,0.7) {$2$};
\node [etichetta]  (33) at (1+3*\ra-0.2,0.7) {$3$};
\node [etichetta]  (34) at (1+4*\ra,0.7) {$4$};
\node [etichetta]  (35) at (1+5*\ra+0.2,0.7) {$5$};
\node [square]  (30) at (1,0) {$5$};
\node [square]  (31) at (1+\ra,0) {$1$};
\node [square]  (32) at (1+2*\ra,0) {$2$};
\node [square]  (33) at (1+3*\ra,0) {$4$};
\node [square]  (34) at (1+4*\ra,0) {$4$};
\node [square]  (35) at (1+5*\ra,0) {$\mathbf{3}$};
\draw [black,->,thick,-latex] (33.north) .. controls (1+3*\ra,1) and (1+5*\ra-0.2,1) .. (35.north);

\end{tikzpicture}
\caption{Removing bucket 0, 3, and 5 in order}
\label{fig:remove 0,3,5 in order}
\end{figure}
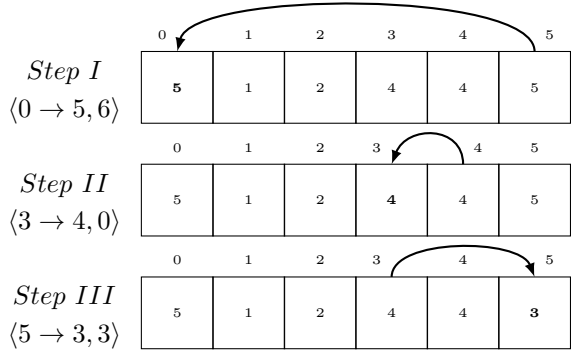

The first operation of the lookup $jump(key,6)$ returns $3$. Bucket $3$ was removed (i.e., $\langle 3 \rightarrow 4,0 \rangle \in \mathcal{R}$, see step II in Fig.~\ref{fig:remove 0,3,5 in order}), so we rehash the key to get an index in the range $[0,4\!-\!1]$ (Alg.~\ref{alg:lookup} lines 3 to 6. Fig.~\ref{fig:lookup-step1}).

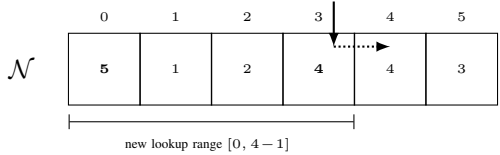
\begin{figure}[H]
\centering
\begin{tikzpicture}[scale=1]
\tikzstyle{every node}=[font=\tiny]
\node [etichetta]  (0) at (-.1,1) {\normalsize $\mathcal{N}$};
\node [etichetta]  (0) at (1,1.7) {$0$};
\node [etichetta]  (1) at (1+\ra,1.7) {$1$};
\node [etichetta]  (2) at (1+2*\ra,1.7) {$2$};
\node [etichetta]  (3) at (1+3*\ra,1.7) {$3$};
\node [etichetta]  (4) at (1+4*\ra,1.7) {$4$};
\node [etichetta]  (5) at (1+5*\ra,1.7) {$5$};
\node [square]  (0) at (1,1.0) {$\mathbf{5}$};
\node [square]  (1) at (1+\ra,1.0) {$1$};
\node [square]  (2) at (1+2*\ra,1.0) {$2$};
\node [square]  (3) at (1+3*\ra,1.0) {$\mathbf{4}$};
\node [square]  (4) at (1+4*\ra,1.0) {$4$};
\node [square]  (5b) at (1+5*\ra,1.0) {$3$};
\draw [black,->,thick,-latex] (1+3*\ra+0.2,1.9) -- (1+3*\ra+0.2,1.3);
\draw [black,densely dotted,->,-latex,thick] (1+3*\ra+0.2,1.3) -- (1+4*\ra,1.3);
\draw [|-|] (1-.5*\ra,0.3) -- (1+3.5*\ra,0.3);
\node [etichetta] (0) at (2.5*\ra,0) {new lookup range $[0,4\!-\!1]$};
\end{tikzpicture}
\caption{First lookup step ($jump(key,6) \rightarrow 3$)}
\label{fig:lookup-step1}
\end{figure}

We assume to use a uniform hash function, so every index has the same probability of being selected. Buckets $0$ and $3$ are removed, and both redirect to bucket $4$ ($3\rightarrow 4$ and $0\rightarrow 5\rightarrow 3\rightarrow 4$). Therefore, buckets $1$ and $2$ have $\frac{1}{4}$ probability of being selected, while the probability of selecting bucket $4$ (which is at the end of the replacement chain) is $\frac{1}{2}$.
Thanks to the condition $u \geq w_b$ of the internal loop, this balance issue is avoided. If the first rehash returns $1$ or $2$, we find a working bucket, and the loop exits. If it returns $3$, we follow the redirect hitting bucket $4$ that is also working, and the loop exits. Finally, if it returns $0$, the internal loop follows the first redirect $\langle 0 \rightarrow 5,3 \rangle \in \mathcal{R}$, hitting $5$. Bucket $5$ was removed after bucket $3$ (i.e., $\langle 5 \rightarrow 3,0 \rangle \in \mathcal{R}$, see step III in Fig.~\ref{fig:remove 0,3,5 in order}), but the condition $u \geq w_b$ does not hold because $u(=3)<w_b(=4)$. Therefore, the internal loop stops, $b$ is set to $5$ (Alg.~\ref{alg:lookup} line 10), $w_b$ is set to $3$ (Alg.~\ref{alg:lookup} line 4), and we rehash the key again to get an index in the interval $[0,3\!-\!1]$ (Fig.~\ref{fig:lookup-step2}).

\begin{figure}[H]
\centering
\begin{tikzpicture}[scale=1]
\tikzstyle{every node}=[font=\tiny]
\node [etichetta]  (0) at (-.1,1) {\normalsize $\mathcal{N}$};
\node [etichetta]  (0) at (1,1.7) {$0$};
\node [etichetta]  (1) at (1+\ra,1.7) {$1$};
\node [etichetta]  (2) at (1+2*\ra,1.7) {$2$};
\node [etichetta]  (3) at (1+3*\ra,1.7) {$3$};
\node [etichetta]  (4) at (1+4*\ra,1.7) {$4$};
\node [etichetta]  (5) at (1+5*\ra,1.7) {$5$};
\node [square]  (0) at (1,1.0) {$\mathbf{5}$};
\node [square]  (1) at (1+\ra,1.0) {$1$};
\node [square]  (2) at (1+2*\ra,1.0) {$2$};
\node [square]  (3) at (1+3*\ra,1.0) {$\mathbf{4}$};
\node [square]  (4) at (1+4*\ra,1.0) {$4$};
\node [square]  (5b) at (1+5*\ra,1.0) {$\mathbf{3}$};
\draw [black,->,thick,-latex] (1+0*\ra+0.2,1.9) -- (1+0*\ra+0.2,1.3);
\draw [black,densely dotted,->,-latex,thick] (1+0*\ra+0.2,1.3) -- (1+5*\ra,1.3);
\draw [|-|] (1-.5*\ra,0.3) -- (1+2.5*\ra,0.3);
\node [etichetta] (0) at (2.0*\ra,0) {new lookup range $[0,3\!-\!1]$};
\end{tikzpicture}
\caption{Second lookup step when $d = 0$}
\label{fig:lookup-step2}
\end{figure}
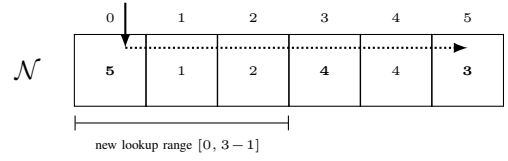

Again, the indexes $0$, $1$, and $2$ have the same probability of being selected. This time, $w_b = 3$ so, if we hit $0$ again, the internal loop keeps following the substitution chain ($0\rightarrow 5\rightarrow 3\rightarrow 4$) until hitting bucket $4$. This allows keys to be distributed evenly among the working buckets. Keys returning the indexes $1$, $2$, or $3$ end the external loop in one iteration, while keys returning the index $0$ do an additional iteration that distributes the keys evenly among the working buckets, as shown in Fig~\ref{fig:lookup-stepTree}. Kindly observe that during each cycle of the external loop, we either identify a working bucket or the lookup range diminishes. 

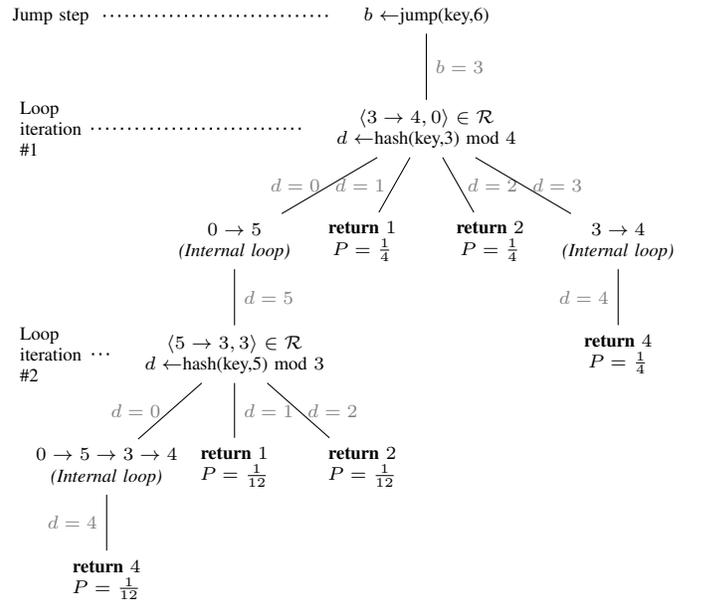
\begin{figure}[H]
\centering
\begin{tikzpicture}[level distance=1.5cm,
  level 1/.style={sibling distance=1cm},
  level 2/.style={sibling distance=1.7cm},
  level 3/.style={sibling distance=1.7cm}]

  \node (l1)[font=\scriptsize]{$b\leftarrow$jump(key,6)}
    child {
      node (l2) [align=center,font=\scriptsize] {
        $\langle 3\rightarrow 4,0\rangle\in\mathcal{R}$\\
        $d\leftarrow$hash(key,3) mod $4$
      }
      child {
        node [align=center,font=\scriptsize] {
          $0\rightarrow 5$\\
          \textit{(Internal loop)}
        }
        child{
          node (l3) [align=center,font=\scriptsize] {
            $\langle 5\rightarrow 3,3\rangle\in\mathcal{R}$\\
            $d\leftarrow$hash(key,5) mod $3$
          }
          child {
            node [align=center,font=\scriptsize] {
              $0\rightarrow 5\rightarrow 3\rightarrow 4$\\
              \textit{(Internal loop)}
            }
            child {
              node [align=center,font=\scriptsize] {
                \textbf{return} $4$\\
                $P=\frac{1}{12}$
              }
              edge from parent node[left,font=\scriptsize,color=gray] {$d=4$}
            }
            edge from parent node[left,font=\scriptsize,color=gray] {$d=0$}  
          }
          child {
            node [align=center,font=\scriptsize] {
              \textbf{return} $1$\\
              $P=\frac{1}{12}$
            }
            edge from parent node[right,font=\scriptsize,color=gray] {$d=1$}
          }
          child {
            node [align=center,font=\scriptsize] {
              \textbf{return} $2$\\
              $P=\frac{1}{12}$
            }
            edge from parent node[right,font=\scriptsize,color=gray] {$d=2$}
          }
          edge from parent node[right,font=\scriptsize,color=gray] {$d=5$}
        }
        edge from parent node[left,font=\scriptsize,color=gray] {$d=0$}
      }
      child {
        node [align=center,font=\scriptsize] {
            \textbf{return} $1$\\
            $P=\frac{1}{4}$
        }
        edge from parent node[left,font=\scriptsize,color=gray] {$d=1$}
      }
      child {
        node [align=center,font=\scriptsize] {
            \textbf{return} $2$\\
            $P=\frac{1}{4}$
        }
        edge from parent node[right,font=\scriptsize,color=gray] {$d=2$}
      }
      child {
        node [align=center,font=\scriptsize]{
          $3\rightarrow 4$\\
          \textit{(Internal loop)}
        }
        child {
          node [align=center,font=\scriptsize] {
            \textbf{return} $4$\\
            $P=\frac{1}{4}$
          }
          edge from parent node[left,font=\scriptsize,color=gray] {$d=4$}
        }
        edge from parent node[right,font=\scriptsize,color=gray] {$d=3$}
      }
    edge from parent node[right,font=\scriptsize,color=gray] {$b=3$}
  };  
  
  \node[left=3.4 of l1]  (ln1)[align=left,font=\scriptsize] {Jump step}
    child {
        node (ln2)[align=left, font=\scriptsize] {Loop \\ iteration \\ \#1}
        edge from parent [draw=none]
        child {
            node (ln3) {}
            edge from parent [draw=none]
            child {
                node (ln4) [align=left,font=\scriptsize]{
                    Loop \\
                    iteration \\ \#2}
                edge from parent [draw=none]
                }}};

    \draw[black,dotted,thick]    
    ($(l1.west)+(-1em,0)$) -- (ln1);
\draw[black,dotted,thick]    
    ($(l2.west)+(-1em,0)$) -- (ln2.east);
\draw[black,dotted,thick]    
    ($(l3.west)+(-1em,0)$) -- (ln4);   

\end{tikzpicture}
\caption{External loop iterations and output probabilities}
\label{fig:lookup-stepTree}
\end{figure}

\begin{proposition}
The lookup function always ends.
\end{proposition}
\begin{proof}
To prove the termination of the lookup function, we must prove that the two nested loops always perform a bounded number of iterations. We will start proving the termination of the internal loop.\\
In section \ref{sec:MementoHash}, it is explained that when a bucket is removed, a substitution in the form $\langle b \rightarrow c,p\rangle$ is added to $\mathcal{R}$, where $b$ is the removed 
bucket, $c$ is the replacing bucket, and $p$ is the bucket previously removed. By construction, either $c$ is working, or there must exist a substitution $\langle c \rightarrow u,q\rangle \in 
\mathcal{R}$. This ensures that a chain of substitutions always ends with a working bucket.
In Algorithm 4, the internal loop (lines 7 to 9) follows this substitution chain until it either reaches a working bucket (the end of the chain) or encounters a bucket $d$ that was removed after $b$
(i.e., a bucket that was working when $b$ was removed). This guarantees the termination of the internal loop. \\
The external loop (lines 3 to 11) begins by selecting a bucket $d$ from the range $[0,n\!-\!1]$. For each iteration of the external loop, the internal loop will return either $d$ if it is working or a bucket $u$ that was working when $d$ was removed (and can still be working). If the internal loop returns $d$ or $u$ as
working, the external loop terminates.
If not, the external loop selects a new $d$ from the range $[0, u\!-\!1]$.  Proposition \ref{prop:lastRepacingIsNewWorkingSize} emphasizes that the value $u$ represents both the replacing bucket and the count of working buckets after the removal of $d$. Hence, $u$ is strictly less than $n$. Considering these factors, it becomes evident that, in each iteration, we either discover a working bucket
or reduce the range $[0, u\!-\!1]$ to $[0, u\!-\!\delta]$, with $\delta > 1$. Ultimately, this range contracts to $[0, w\!-\!1]$, where each position leads to a working bucket. Therefore, the external loop is guaranteed to terminate.
\end{proof}

\begin{proposition} \label{prop:minimalDisruption}
Minimal disruption holds for Memento.
\end{proposition}
\begin{proof}
When we remove a bucket $b$, only the keys previously mapped to $b$ should move to a new location. We will prove this property by induction on $r$ (the number of removed buckets).
\begin{itemize}
\item \textbf{r = 1}: If we remove the bucket from the tail, minimal disruption is guaranteed by \textit{Jump}. Otherwise, let us assume to remove a bucket $b$ other than the last. First, we invoke \textit{Jump}. if it returns an index other than $b$ the algorithm ends. The position of the keys previously mapped to a bucket other than $b$ does not change. If \textit{Jump} returns $b$ the key is rehashed and mapped to an index in the range [0,n-2]. If it hits $b$ again, it will be redirected to $c \neq b$. Only keys previously mapped to $b$ will move to another bucket.
\item \textbf{r}: we assume the property to hold for $r \geq 1$.
\item \textbf{r+1}: by hypothesis all the keys were consistently moved to the $n-r$ working buckets. After we remove a new bucket $b$, the lookup algorithm stops if a key hits a working bucket. Keys on working buckets do not move. Keys hitting a previously removed bucket will not move by hypothesis. Finally, if a key hits $b$ (the last removed bucket), it will be rehashed and mapped to an index in the range $[0,n\!-\!r\!-\!2]$ that contains only working buckets. Therefore, only keys mapped to $b$ will move to a working bucket.
\end{itemize}
\end{proof}

\begin{proposition} \label{prop:balance}
Balance holds for Memento
\end{proposition}
\begin{proof}
We prove the balance by induction on the number of removed buckets $r$.
\begin{itemize}
\item \textbf{r = 0}: balance is guaranteed by \textit{Jump}.
\item \textbf{r = 1}: first we invoke \textit{Jump}, evenly distributing keys among all the buckets. Every bucket gets $\frac{k}{n}$ keys. The keys mapped on the only removed bucket $b$ are rehashed and mapped to an index in the range $[0,n\!-\!2]$. The hash function is assumed to be uniform, so every index gets $\frac{\frac{k}{n}}{n-1}$ keys. The keys that hit $b$ again are remapped to $c$. Therefore, all the buckets except $b$ will end up having $\frac{k}{n}+\frac{\frac{k}{n}}{n-1}=\frac{k}{n}+\frac{k}{n(n-1)}=\frac{k(n-1)+k}{n(n-1)}=\frac{kn}{n(n-1)}=\frac{k}{n-1}$ keys.
\item \textbf{r}: we assume the property to hold for $r \geq 0$.
\item \textbf{r+1}: we assume the keys to be evenly distributed after removing $r$ buckets. Every working bucket has $\frac{k}{n-r}$ keys. When we remove a new bucket $b$, all the keys previously mapped to other buckets will not be affected (minimal disruption). The keys previously mapped to $b$ will be rehashed and mapped to an index in the range $[0,n-r-2]$. Since we use a uniform hash function, every index in $[0,n\!-\!r\!-\!2]$ gets $\frac{1}{n\!-\!r\!-\!1}$ of the rehashed keys. After the removal of $b$ every working bucket will have $\frac{k}{n-r}+\frac{\frac{k}{n-r}}{n-r-1}=\frac{k}{n-r}+\frac{k}{(n-r-1)(n-r)}=\frac{k(n-r)}{(n-r)(n-r-1)}=\frac{k}{n-r-1}$ keys.
\end{itemize}
\end{proof}

\begin{proposition}
Monotonicity holds for Memento.
\end{proposition}
\begin{proof}
When a new bucket is added, keys only move from an existing bucket to the new one, but not between existing ones. We will prove this property by induction on the number of removed buckets $r$.
\begin{itemize}
\item \textbf{r = 0}: there are no removed buckets. The new bucket is added to the tail and monotonicity is guaranteed by \textit{Jump}.
\item \textbf{r = 1}: we have only one removed bucket $b$. Keys that are sent to $b$ by \textit{Jump} in the first place are evenly remapped to the remaining buckets. As proven in \ref{prop:balance}, every bucket other than $b$ gets $\frac{1}{n-1}$ of the keys previously mapped to $b$. When we add a new node to the cluster, \textit{Memento} will assign to such a node the bucket $b$. Keys previously mapped to $b$ will move to $b$ again, so from every node other than $b$, $\frac{k}{n(n-1)}$ keys will move to $b$. A total of $\frac{k}{n}$ keys will move to $b$. Keys mapped by \textit{Jump} to buckets other than $b$ are not affected by this operation and, therefore, will not move.
\item \textbf{r}: we assume the property to hold for $r \geq 1$.
\item \textbf{r+1}: by hypothesis monotonicity holds when up to $r$ buckets are removed. We assume to have $r+1$ removed buckets and add a new node to the cluster. \textit{Memento} restores the last removed bucket. As discussed in \ref{prop:minimalDisruption}, keys were remapped to the working buckets when the last bucket $b$ was removed. Each working bucket got $\frac{1}{n-r-1}$ of the keys previously mapped to $b$. When $b$ gets restored, keys are not remapped anymore. Therefore, any working bucket other than $b$ will send $\frac{k}{(n-r)(n-r-1)}$ keys back to $b$. A total of $\frac{k}{n-r}$ keys will move to $b$, and no keys will move to any other node.
\end{itemize}
\end{proof}

\section{Computational complexity}
In this section, we explore the computational complexity of the functions described earlier. The complexity related to initializing the data structure, adding and removing buckets is considered negligible, thus requiring only a brief explanation. For the lookup function we provide instead a formal characterisation to elaborate on its intricate nature.

\subsection{Initialization}
During the initialization phase, we assume all nodes to be working and each bucket to match its own position. The set of replacements is initialized as an empty set (or hash table) and the size of the b-array is stored in the algorithm state. The initialization takes therefore $O(1)$ and uses minimal memory.

\subsection{Removing a bucket}
If $b=n-1$ and $\mathcal{R}=\emptyset$ the algorithm just updates the size of the b-array which takes $O(1)$.
Otherwise, a new replacement is created and added to the replacement set. The replacement is a tuple and can be created in constant time. The replacement can be added in constant time by using a hash table to represent $\mathcal{R}$. So, the overall complexity of removing a bucket is $O(1)$.

\subsection{Adding a bucket}
If $\mathcal{R}=\emptyset$ (no previously removed buckets), we assume the bucket to be added to the tail of the b-array. In this case, the algorithm stores such information by increasing the size of the b-array, which takes $O(1)$.
Otherwise, the algorithm restores the last removed bucket. A replacement can be removed in constant time by using a hash table to represent $\mathcal{R}$. So, the overall complexity of adding a bucket is $O(1)$.

\subsection{Lookup}
In this section, we derive an upper bound to the average number of operations needed to lookup a key. The algorithm starts invoking $b\gets jump(key,n)$ which takes $O(\ln(n))$ \cite{lamping2014fast}. If the first step hits a working bucket the algorithm is done. Otherwise there are two nested loops. We will start by giving an upper bound to the number of iterations of the external loop (Alg.~\ref{alg:lookup} lines 3-11).

\begin{proposition}[\textbf{External loop computational complexity}]\label{prop:outer}
Let $\tau$ denote the number of iterations of the external loop in Alg.~\ref{alg:lookup} given a b-array of size $n$ with $w \leq n$ working buckets and a given key. Due to the stochastic nature of the underlying process, $\tau$ is a random variable and we have that: 
\begin{itemize}
\item[(i)] the expectation of $\tau$ is upper-bounded by $ln(\frac{n}{w})$; and 
\item[(ii)] the standard deviation of $\tau$ is upper-bounded by $\sqrt{ln(\frac{n}{w})}$.
\end{itemize}
\end{proposition}

\begin{proof}
We prove the proposition by deriving a closed-form expression for the \emph{moment generating function} (MGF, \cite{cressie1986moment}) of $\tau$, that is defined as:
\begin{equation}
\phi_\tau(s) := \mathbb{E}[e^{s\tau}]\,,
\end{equation}
with $\mathbb{E}$ denoting the expectation of its argument computed w.r.t. the distribution $\mathbb{P}(\tau)$. It is a straightforward exercise to check that the expectation and the variance of $\tau$ can be expressed by means of the first two derivatives of the MGF computed for $s=0$:
\begin{eqnarray}\label{eq: first moment}
\mathbb{E}[\tau] &=& \phi'_\tau(0)\,,\\
\mathbb{V}(\tau) &=& \phi_\tau''(0)-(\phi_\tau'(0))^2\,.
\label{eq: standard deviation}
\end{eqnarray}

Remember that $\mathcal{W}_b$ denotes the set of working buckets after $b$ was removed. When the first step (i.e., $jump(key,n)$) hits a removed bucket $b$, there must exist a replacement tuple $\langle b\rightarrow c,p \rangle\in\mathcal{R}$ such that $c\in \mathcal{W}_b$ represents both the replacing bucket (see Prop.~\ref{prop:lastRepacingIsNewWorkingSize}) and the number of working buckets after $b$ was removed (i.e., $c=|\mathcal{W}_b|$). The existence of such a tuple is the entering condition for the external loop. The key is consequently rehashed to an index in the interval $[0,c\!-\!1]$ with uniform probability, and $b$ is replaced by either a working bucket or a bucket removed after $b$. The interval $[0,c\!-\!1]$ becomes smaller after each iteration, and the loop terminates when we reach the interval containing only working buckets.

Consider a fixed sequence of removed buckets $r_{n-w-1} \rightarrow \dots \rightarrow r_0$ where $r_{n-w-1}$ is the first removed bucket and $r_0$ the last one. Denote as $\mathcal{W}_{r_i}$ the set of working buckets after the removal of $r_i$ (hence, $\mathcal{W}_{r_0}=\mathcal{W}$). We might regard the working buckets as buckets that have not been already removed, and extend the sequence by adding working buckets with negative indexes to the tail, i.e., $r_{n-w-1} \rightarrow \dots \rightarrow r_0 \rightarrow r_{-1} \rightarrow \dots \rightarrow r_{-w}$. In this way, every bucket $b$ is equal to $r_i$ for some $i$, being working iff $i<0$. In particular, if $b=jump(key,n)$ is the bucket returned by the first instruction of the algorithm (line 2) there is a $i$ such that $b=r_i$. If $b\in\mathcal{W}$ (i.e., $i<0$) the loop exits and $\tau=1$. If this is not the case, we call $\tau_i$ the number of iterations of the loop when $i\ge 0$. 

When the algorithm enters the loop, a new bucket $d\in\mathcal{W}_b$ is picked with uniform probability. Let us consider a discrete random variable $j$ defined as the index (smaller than $i$ by construction) such that $d=r_j$. If $j<0$, we have that $d$ is working, the loop ends and $\tau_i=1$. Otherwise, if $0\leq j<i$, by the uniform hashing assumption (see Note~\ref{note:uniformHashFunctions}), $\tau_i=1+\tau_j$. Finally, notice that, since $\mathcal{W}_{r_0}=\mathcal{W}$, $\tau_0=1$.

Consider the MGF of $\tau_i$, which, with a small abuse of notation, is denoted as:
\begin{equation} \label{eq:phii}
\phi_i(s):=\mathbb{E}[e^{s\tau_i}]\,.
\end{equation}
This allows us to apply the law of total expectation to $\phi_i(s)$ by conditioning w.r.t. $j$, i.e.,
\begin{equation}\label{eq: phi of i}
\begin{split}
\phi_i(s) = & \mathbb{P}(j<0) \cdot\mathbb{E}[e^{s\tau_i}|j<0]\\
&+\sum_{k=0}^{i-1} \mathbb{P}(j=k) \cdot \mathbb{E}[e^{s\tau_i}|j=k]\,.
\end{split}
\end{equation}
Distribution $\mathbb{P}(j)$ is uniform over its $w+i$ integer values. We can therefore rewrite Eq.~\eqref{eq: phi of i} as:
\begin{equation}\label{eq: phi of i2}
\begin{split}
\phi_i(s)=&\frac{w}{w+i} \cdot \mathbb{E}[e^{s\tau_i}|j<0]\\
&+\frac{1}{w+i} \sum_{k=0}^{i-1} \mathbb{E}[e^{s\tau_i}|j=k]\,.
\end{split}    
\end{equation}
For $j<0$ (i.e., $d\in\mathcal{W}$), as the loop ends after one iteration and a single hash calculation is done, we have $\tau_i=1$ and hence:
\begin{equation}\label{eq: case d in W}
    \mathbb{E}[e^{s\tau_i}|j<0]=e^s\,.
\end{equation}
For each $j=k$ with $0\leq k<i$, we already proved that $\tau_i=1+\tau_k$ and hence:
\begin{equation} \label{eq: case d removed}
\mathbb{E}[e^{s\tau_i}|j=k]=\mathbb{E}[e^{s(1+\tau_k)}]
=e^s\mathbb{E}[e^{s\tau_k}]\,.
\end{equation}
Overall, by Eqs.~(\ref{eq: case d in W}) and (\ref{eq: case d removed})
and the definition in Eq.~\eqref{eq:phii}, we can rewrite Eq.~\eqref{eq: phi of i2} as follows:
\begin{equation}\label{eq: recursive phi}
\phi_i(s)=\frac{w e^s}{w+i}+\frac{e^s}{w+i}\sum_{k=0}^{i-1}\phi_k(s)\,,
\end{equation}
and hence, by simple algebra:
\begin{equation} \label{eq: closed form 1}
\frac{w+i}{e^s}\phi_i(s)=w+\sum_{k=0}^{i-1}\phi_k(s)\,,
\end{equation}
and finally, by putting in evidence $\phi_{i-1}(s)$:
\begin{equation} \label{eq: closed form 1b}
\frac{w+i}{e^s}\phi_i(s)=w+\sum_{k=0}^{i-2}\phi_k(s)+\phi_{i-1}(s)\,.
\end{equation}
If we rewrite Eq.~(\ref{eq: closed form 1}) with $i-1$ instead of $i$, we obtain on the r.h.s. the first two terms of the r.h.s. of Eq.~(\ref{eq: closed form 1b}). The latter equation might therefore rewrite as:
\begin{equation} \label{eq: closed form 2}
 \frac{w+i}{e^s}\phi_i(s)=\frac{w+i-1}{e^s}\phi_{i-1}(s)+\phi_{i-1}(s)\,,
\end{equation} 
and hence, by simple algebra, as the following recursion:
\begin{equation} \label{eq: closed form 3}
\phi_i(s)=\frac{w+i-1+e^s}{w+i}\phi_{i-1}(s)\,.
\end{equation}

As $\phi_0(s)=e^s$, for each $i>0$ we obtain the following closed-form expression for Eq.~(\ref{eq: closed form 3}):
\begin{equation} \label{eq: closed form final}
\phi_i(s)=e^s\prod_{k=1}^i\frac{w+k-1+e^s}{w+k}\,,
\end{equation} 
and hence, by taking the logarithm on both sides:
\begin{equation} \label{eq: closed form log}
\ln(\phi_i(s))=s+\sum_{j=1}^i\ln \left(\frac{w+j-1+e^s}{w+j}\right)\,,
\end{equation}
and finally, by taking the derivative on both sides:
\begin{equation} \label{eq: first derivative}
\frac{\phi_i'(s)}{\phi_i(s)}=1+\sum_{k=1}^i\frac{e^s}{w+k-1+e^s}\,.
\end{equation}
By the MGF definition, we trivially have $\phi_i(0)=1$. Thus, from Eqs.~\eqref{eq: first moment} and \eqref{eq: first derivative}:
\begin{equation}\label{eq:taui}
\mathbb{E}[\tau_i]=\phi_i'(0)=1+\sum_{k=1}^i\frac{1}{w+k}\,.
\end{equation}
Note that the number $\tau$ of iterations for the external loop is stochastically larger when $b=r_{n-w}$ is the first removed bucket. Therefore, we can upper-bound the expectation of $\tau$ and, by Eq.~\eqref{eq:taui}, write:
\begin{equation}
\mathbb{E}[\tau]\leq\mathbb{E}[\tau_{n-w}]=\phi'_{n-w}(0)
=1+\sum_{k=1}^{n-w}\frac{1}{w+k}\,.
\end{equation}
The proof of (i) finally follows from the fact that the sum in the r.h.s. of Eq.~\eqref{eq:taui} is equal to the difference between the $n$-th and th $w$-th \emph{harmonic numbers} being therefore dominated by the difference of the corresponding natural logarithms, i.e.,\footnote{If $H_k$ is the $k$-th harmonic number, then $H_k=\ln k + \gamma_k$, where $\gamma_k$ is a monotonic function of $k$ \cite{sun2012arithmetic}.}
\begin{equation}\label{eq:bound}
1+\sum_{k=1}^{n-w}\frac{1}{w+k} \le 1 + \ln\left(\frac{n}{w}\right)\,.
\end{equation}

We similarly upper-bound the variance (and then the standard deviation), by first taking the derivative of Eq.~(\ref{eq: first derivative}):
\begin{equation} \label{eq: second derivative}
\frac{\phi''_i(s)\phi_i(s)-(\phi'_i(s))^2}{(\phi_i(s))^2}
=\sum_{k=1}^i \frac{e^s(w+k-1+e^s)-e^{2s}}{(w+k-1+e^s)^2}\,.
\end{equation}

Evaluating Eq.~\eqref{eq: second derivative} for $s=0$ and noting that $\phi_i(0)=1$ we obtain on the l.h.s. the variance of $\tau_i$ as in Eq.~\eqref{eq: standard deviation} and hence:
\begin{equation}
\mathbb{V}(\tau_i)=\sum_{k=1}^i \frac{w+k-1}{(w+k)^2}\,.
\end{equation}
As for the expected value the variance of $\tau$ is stochastically dominated by the worst-case value in $b=r_{n-w}$, thus:
\begin{equation}
\mathbb{V}(\tau) \leq \mathbb{V}(\tau_{n-w})=\sum_{k=1}^{n-w}\frac{w+k-1}{(w+k)^2}\,,
\end{equation}
and hence:
\begin{equation}
\sum_{k=1}^{n-w}\frac{w+k-1}{(w+k)^2}
\leq\sum_{k=1}^{n-w}\frac{w+k}{(w+k)^2}\leq\ln{(\frac{n}{w})}\,,
\end{equation}
where the first inequality follows by simple algebra and the second is because of Eq.~\eqref{eq:bound}. As the standard deviation is the square root of the variance, it is upper-bounded by $\sqrt{\ln(\frac{n}{w})}$. Note that the bounds do not depend on the removal sequence.
\end{proof}

Let us similarly characterise the worst-case complexity of the internal loop of the algorithm by the following result.
\begin{proposition}[\textbf{Internal loop computational complexity}]
Let $\sigma$ denote the number of iterations of the internal loop in Alg.~\ref{alg:lookup} given a b-array of size $n$ with $w \leq n$ working buckets and a given key. Due to the stochastic nature of the underlying process, $\sigma$ is a random variable and we have that:
\begin{itemize}
\item the expectation of $\sigma$ is upper-bounded by $ln(\frac{n}{w})$; and
\item the standard deviation of $\sigma$ is upper-bounded by $\sqrt{ln(\frac{n}{w})}$.
\end{itemize}
\end{proposition}

\begin{proof}
The proof is very similar to that of Prop.~\ref{prop:outer}. As in the other proof, let us extend the sequence of removals by adding the working buckets to the tail with negative indexes: $r_{n-w-1} \rightarrow \dots \rightarrow r_0 \rightarrow r_{-1} \rightarrow \dots \rightarrow r_{-w}$.

The first step of the algorithm returns a random bucket $b\gets jump(key,n)$ with uniform probability.
But, if we enter the external loop, then $b$ is a removed bucket. Let us consider a discrete random variable $x$ defined as the index such that $b=r_x$. As $b$ is a removed bucket, we have $x \ge 0$.
In this case, the algorithm generates a new bucket $d=r_i$ for some $i$. The condition to enter also the internal loop is $\exists \langle d \rightarrow u,q \rangle \in \mathcal{R}$ and $u \geq w_b$ can be translated as follows:
\begin{quote}
{\emph{The loop ends if $d$ is a working bucket or if $d$ was removed after $b$, i.e., if $b=r_x$, $d=r_i$, and $i<x$. The loop executes instead a new iteration if $i\ge x$.}}
\end{quote}
Let $\sigma_i$ denote the number of iterations of the internal loop when $d=r_i$. If $d$ is working or was removed after $b$ (i.e., $i<x$), the loop ends immediately and $\sigma_i=1$. Otherwise, $d$ is replaced by $u$ (Alg.~\ref{alg:lookup} line 8) with $u=r_j$ for some $j<i$. Since the sequence of removed buckets can be any, also the substitutions $d\rightarrow u$ are randomly defined. Therefore, when $i > x$ the distribution of $\sigma_i$ follows the same distribution as $1+\sigma_j$.
We can therefore proceed as in Eq.~\eqref{eq: phi of i} and write the MGF of $\sigma_i$ as:
\begin{equation} \label{eq: internal loop random x}
\begin{split}
    \mathbb{E}[e^{t\sigma_i}]=&\mathbb{P}(j<x)\cdot\mathbb{E}[e^{t\sigma_i}|j<x]\\
    +&\sum_{k=x}^{i-1}\mathbb{P}(j=k)\cdot\mathbb{E}[e^{t\sigma_{i}}|j=k]\,.
\end{split}
\end{equation}
The number of iterations of the internal loop is stochastically larger when $b=r_0$ is the last removed bucket (i.e., $x=0$). In this case, all the other buckets were removed before $b$. Thus, replacing $x$ with $0$ in Eq.~(\ref{eq: internal loop random x}) holds:
\begin{equation} \label{eq: internal loop with x=0}
\begin{split}
    \mathbb{E}[e^{t\sigma_i}]\le& \mathbb{P}(j<0) \cdot \mathbb{E}[e^{t\sigma_i}|j<0]\\
    &+\sum_{k=0}^{i-1}\mathbb{P}(j=k) \cdot \mathbb{E}[e^{t\sigma_{i}}|j=k]\,,
\end{split}
\end{equation}
and finally:
\begin{equation} \label{eq: internal loop with x=02}
\begin{split}
\mathbb{E}[e^{t\sigma_i}] &\le 
\frac{w}{w+i}\cdot \mathbb{E}[e^{t\sigma_i}|j<0]\\&+\frac{1}{w+i}\sum_{k=0}^{i-1}\mathbb{E}[e^{t\sigma_i}|j=k]\,.
\end{split}
\end{equation}

Since the r.h.s. of Eq.~(\ref{eq: internal loop with x=02}) is the same as in Eq.~(\ref{eq: phi of i2}) and the same conditions hold, we can apply the same steps as for the external loop to obtain the proof of the proposition.
\end{proof}

Finally, we can provide an upper bound to the overall complexity of the two nested loops.

\begin{proposition}[\textbf{Nested loops computational complexity}]
Let $\omega$ denote the overall number of iterations executed by the nested loops in Alg.~\ref{alg:lookup}. Due to the stochastic nature of the underlying process, $\omega$ is a random variable and we have that:
\begin{itemize}
\item[(i)] the expectation of $\omega$ is upper-bounded by $\left[ln(\frac{n}{w})\right]^2$; and
\item[(ii)] the standard deviation of $\omega$ is upper-bounded by $\left[ln(\frac{n}{w})\right]^{\frac{3}{2}}$.
\end{itemize}
\end{proposition}
\begin{proof}
Since the two loops are nested, the total number of iterations is given by the product of the iterations of the individual loops.
Hence, $\omega = \tau \cdot \sigma$. As those two random variables are independent the expected values factorise, i.e.,
\begin{equation}
\mathbb{E}[\tau\cdot\sigma]=\mathbb{E}[\tau]\cdot\mathbb{E}[\sigma]\,,
\end{equation}
and (i) easily follow from the results about the expectations in the two previous propositions.

Finally, to characterise the standard deviation, we first consider a classical characterisation of the variance of the product of two independent variables such as $\tau$ and $\sigma$:
\begin{equation}
\mathbb{V}(\tau \cdot \sigma) = \mathbb{V}(\tau) \cdot \mathbb{V}(\sigma) + \mathbb{V}(\tau) \cdot \mathbb{E}[\sigma]^2 + \mathbb{E}[\tau]^2 \cdot \mathbb{V}(\sigma)\,,
\end{equation}
and hence, by considerations analogous to those provided for the expectations:
\begin{equation}
\mathbb{V}(\tau \cdot \sigma) \leq \left[ \ln \left(\frac{n}{w}\right) \right]^3\,,
\end{equation}
and thence the upper bound in (ii) for the standard deviation.
\end{proof}

The overall complexity of the \emph{lookup} function is given by summing the complexity of \textit{Jump}, invoked as the first instruction, and the complexity of the nested loops, resulting in $O(ln(n)+[ln(\frac{n}{w})]^2)$.
The following table summarizes the asymptotic complexity in time and space of \textit{Memento} and the other algorithms considered in the paper.
\begin{table}[H]
  \centering
  \caption{Asymptotic complexity}
  \label{table:complexity}
  \begin{threeparttable}
    \begin{tabular}{|l|c|c|c|c|}
      \hline
      & Memento & Jump \\
      \hline
      Memory usage & $\Theta(r)$ & $\Theta(1)$ \\
      Lookup time & $O(ln(n)+[ln(\frac{n}{w})]^2)$ & $O(ln(w))$ \\
      Init time & $\Theta(1)$ & $\Theta(1)$ \\
      Resize time & $\Theta(1)$ & $\Theta(1)$ \\
      \hline
      & Anchor & Dx \\
      \hline
      Memory usage & $\Theta(a)$ & $\Theta(a)$ \\
      Lookup time & $O([ln(\frac{a}{w})]^2)$ & $O(\frac{a}{w})$ \\
      Init time & $\Theta(a)$ & $\Theta(a)$ \\
      Resize time & $\Theta(1)$ & $\Theta(1)$ \\
      \hline
    \end{tabular}
    \begin{tablenotes}
      \footnotesize
      \item[] w = number of working buckets
      \item[] r = number of removed buckets
      \item[] a = overall capacity of the cluster
      \item[] n = size of the b-array
    \end{tablenotes}
  \end{threeparttable}
\end{table}

\section{Benchmarks}
\label{sec:benchmarks}
We conducted a benchmark comparing the performance of \textit{Memento} with \textit{Jump}, \textit{Dx}, and the in-place version of \textit{Anchor} in terms of lookup time and memory usage. The implementations of these algorithms, together with the benchmarking tool, can be found on GitHub \cite{isinGitHub}.
Our findings indicate that \textit{Memento} performs similarly to \textit{Jump} in the best-case scenario and better than \textit{Anchor} and \textit{Dx} in the worst-case scenario (up to $65\%$ of removed nodes). As mentioned in Sec.~\ref{preliminaries}, \textit{Anchor} and \textit{Dx} require setting an upper bound for the cluster's overall capacity (referred to as the value $a$ in Tab.~\ref{table:complexity}). This value impacts both time and memory complexity. Determining the appropriate overall capacity is not a straightforward task. Setting the overall capacity  to the initial cluster size (i.e., $a=w$), or merely doubling the initial cluster size (i.e., $a=2*w$) is unrealistic because real-world scenarios often involve scaling clusters up to tens or hundreds of times the initial size, such as in distributed storage systems. Therefore, we believe that the overall capacity of the cluster should be much bigger than the initial cluster size (i.e., $a >> w$). However, using a significantly larger overall capacity would negatively affect the performance of these algorithms. Hence, we decided to settle for a reasonable compromise that reflects a real-world setup without excessively penalizing performance, setting the overall capacity to be ten times the initial capacity (i.e., $a = 10*w$). In Sec.~\ref{sec:sensitivity} we provide a sensitivity analysis with respect to ratio $\frac{a}{w}$.

\subsection{Evaluation scenarios}
We propose a comprehensive set of empiric evaluation scenarios to determine the behavior of our algorithm. The considered metrics are \textit{lookup time} (i.e., the time required for resolving the bucket which stores a particular key), and \textit{memory usage}. We experimented with three main scenarios:
\begin{itemize}
    \item \textit{Stable}: in this scenarios the nodes (buckets) of the network are stable and no additions or removals are performed (the size of the cluster remains $w$ throughout the experiment);
    \item \textit{One-shot removals}: given an initial size (number of nodes), we remove $90\%$ of the initial nodes ($w$);
    \item \textit{Incremental removals}: starting from an initial size of one million nodes ($w=10^6$), a growing number of nodes (buckets) is removed from the network. The amount of nodes removed will be expressed as a percentage of the initial size ($w$).
\end{itemize}

For every benchmark, we always initialize \textit{Anchor} and \textit{Dx} with $a = 10*w$.
When the network is constructed we assume that nodes are inserted, one at a time, into a corresponding data structure in a known order. This ordering influences the performance of the algorithms, therefore concerning removals, we evaluated two different strategies: for the \textit{best case}, removals follow a Last-In-First-Out (LIFO) order (i.e. the first node to be removed is the last node added to the network), whereas for the \textit{worst case} we remove the nodes in random order. Since \textit{Jump} only supports removals in \textit{LIFO} order, the \textit{worst case} results for that algorithm will also refer to a \textit{LIFO} removal order.

\begin{figure}[ht!]
    \includegraphics[width=0.48\textwidth]{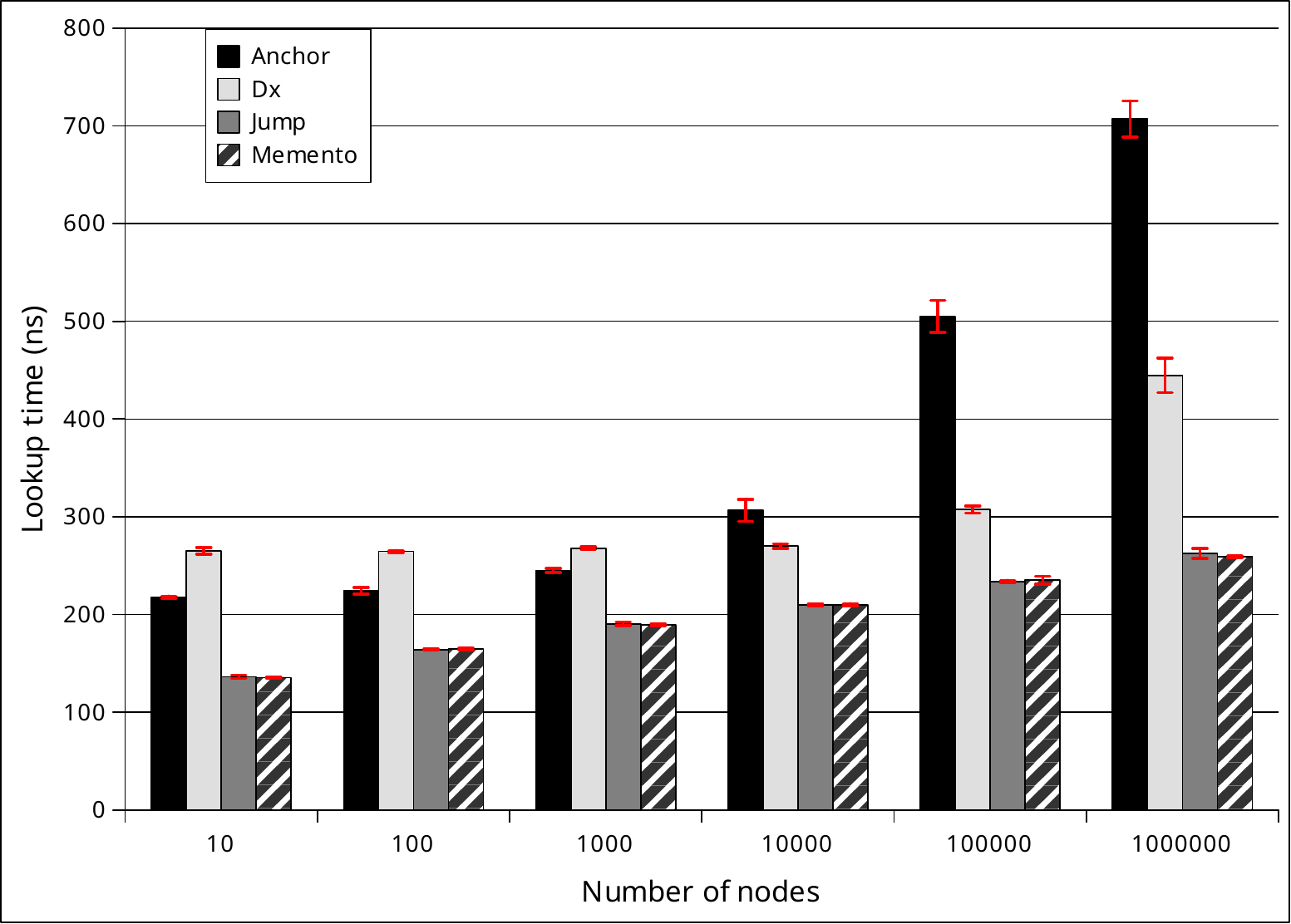}
    \caption{Stable scenario - Lookup time}
    \label{fig:lookup_time_no_removals}
\end{figure}

\begin{figure}[ht!]
    \includegraphics[width=0.48\textwidth]{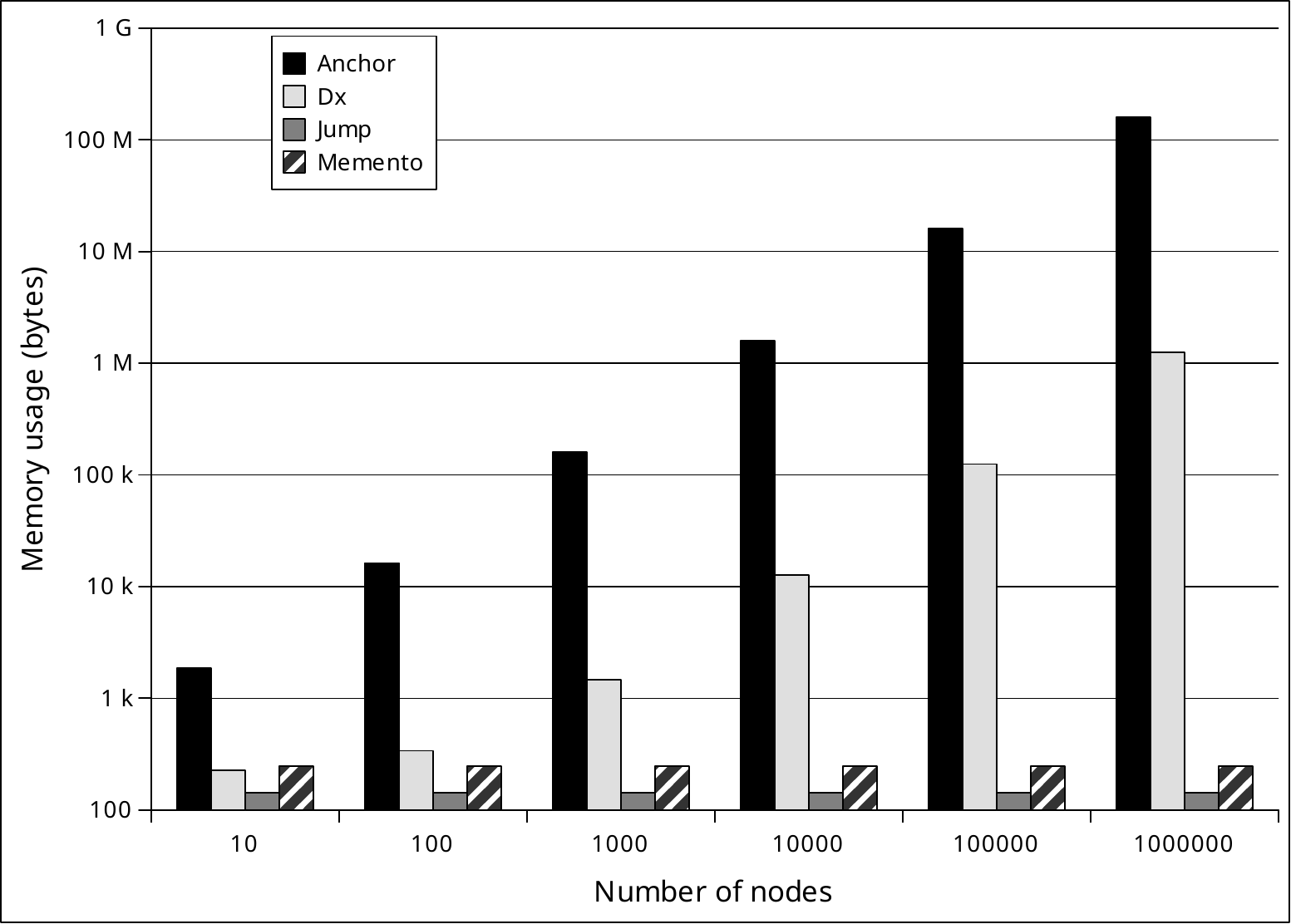}
    \caption{Stable scenario - Memory usage}
    \label{fig:memory_usage_no_removals}
\end{figure}

\subsection{Stable scenario}
In the first evaluation scenario we consider stable networks of different sizes (from ten to one million nodes). Concerning the lookup time, as shown in Fig. \ref{fig:lookup_time_no_removals}, \textit{Memento} performs similarly to \textit{Jump}, and noticeably better than both \textit{Anchor} and {Dx}. Regarding memory usage, \textit{Jump} does clearly benefit from the absence of an internal data structure, which results in the lowest requirement (as shown in Fig. \ref{fig:memory_usage_no_removals} and \ref{fig:memory_usage_wc}). With a constant memory usage, \textit{Memento} is on par with \textit{Jump}, whereas \textit{Anchor} (which has to keep track of the working buckets) and \textit{Dx} (which optimizes memory usage by using a bit-array to mark available buckets) are the worst performers.

\begin{figure}[ht!]
    \includegraphics[width=0.48\textwidth]{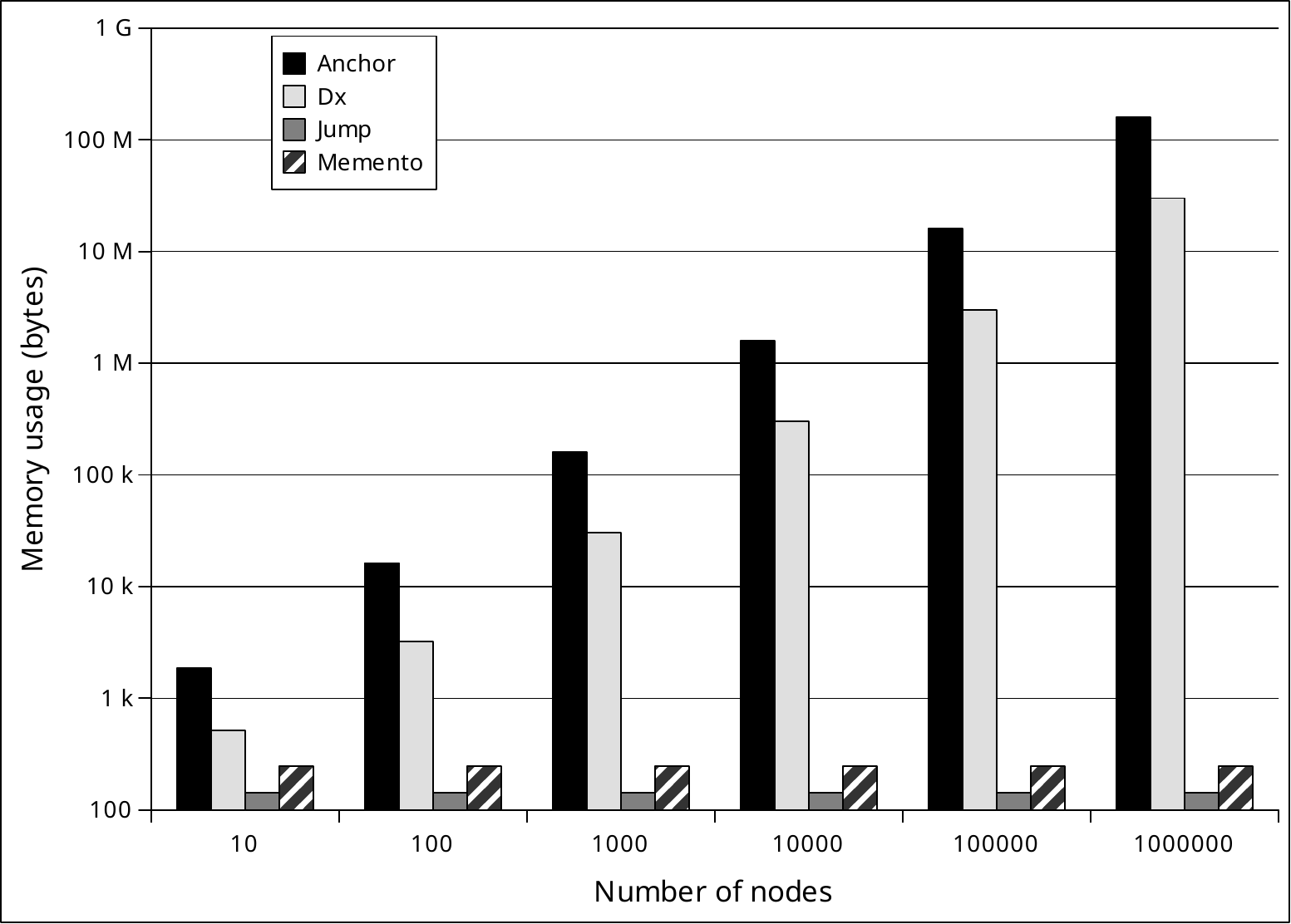}
    \caption{One-shot removals - Memory usage (best case)}
     \label{fig:memory_usage_bc}
\end{figure}

\begin{figure}[ht!]
    \includegraphics[width=0.48\textwidth]{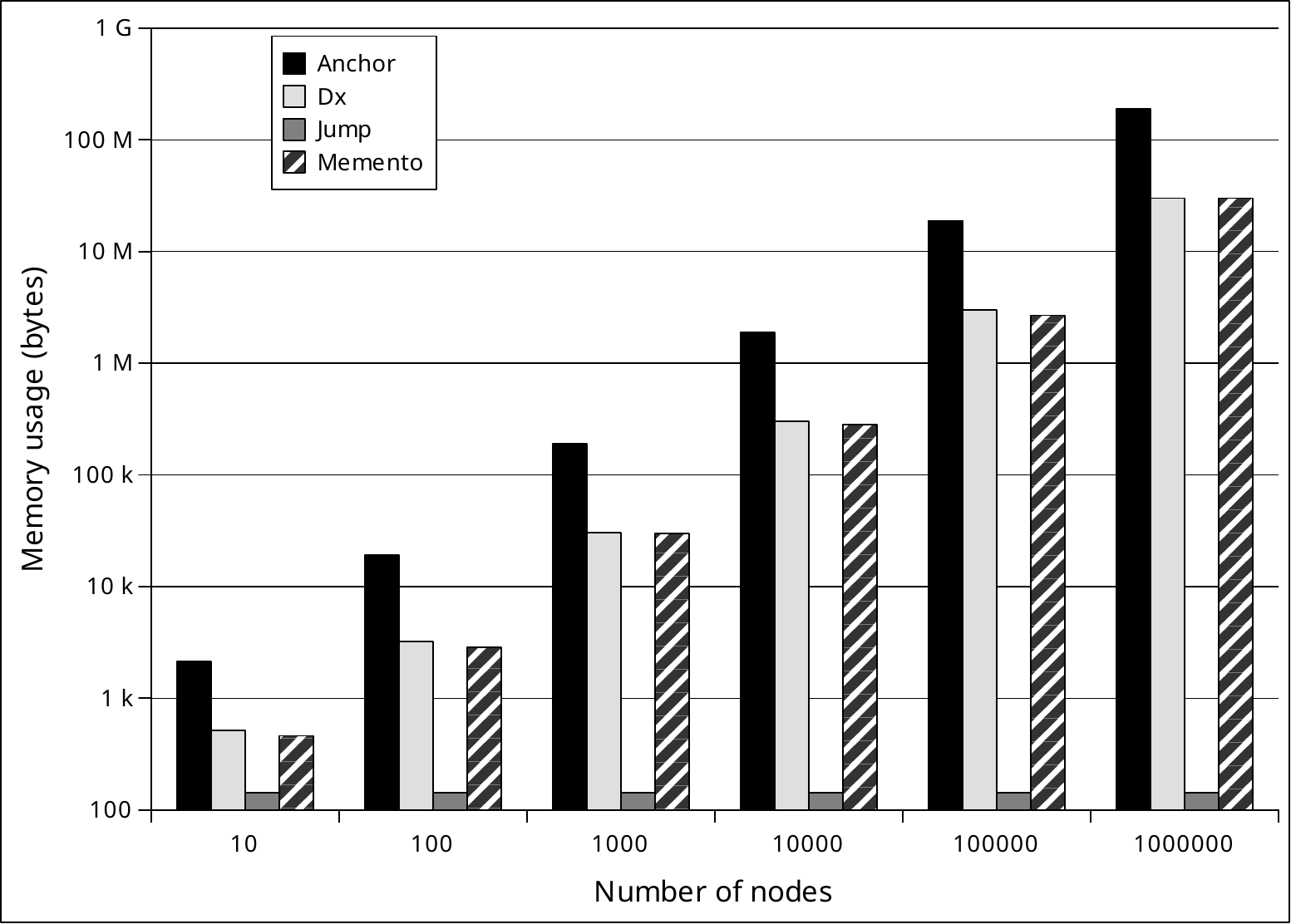}
    \caption{One-shot removals - Memory usage (worst case)}
    \label{fig:memory_usage_wc}
\end{figure}

\subsection{One-shot removals}
When nodes are removed, all algorithms apart from \textit{Jump} need to store additional information into their data structure. This behavior is reflected in an increased memory usage. As shown in Fig.~\ref{fig:memory_usage_bc}, in the best case (namely when nodes are removed in a \textit{LIFO} order), both \textit{Memento} and \textit{Jump} use very little memory and exhibit constant requirements regardless of the size of the network. On the contrary, \textit{Dx} and \textit{Anchor} actively keep track of removed nodes, therefore increasing their memory usage. In the worst-case scenario (nodes removed randomly), \textit{Memento} also needs to update its state and consumed memory increases (Fig.~\ref{fig:memory_usage_wc}). To notice that, even in its worst case scenario, \textit{Memento} needs less memory than \textit{Anchor} and \textit{Dx}.

\begin{figure}[ht!]
    \includegraphics[width=0.48\textwidth]{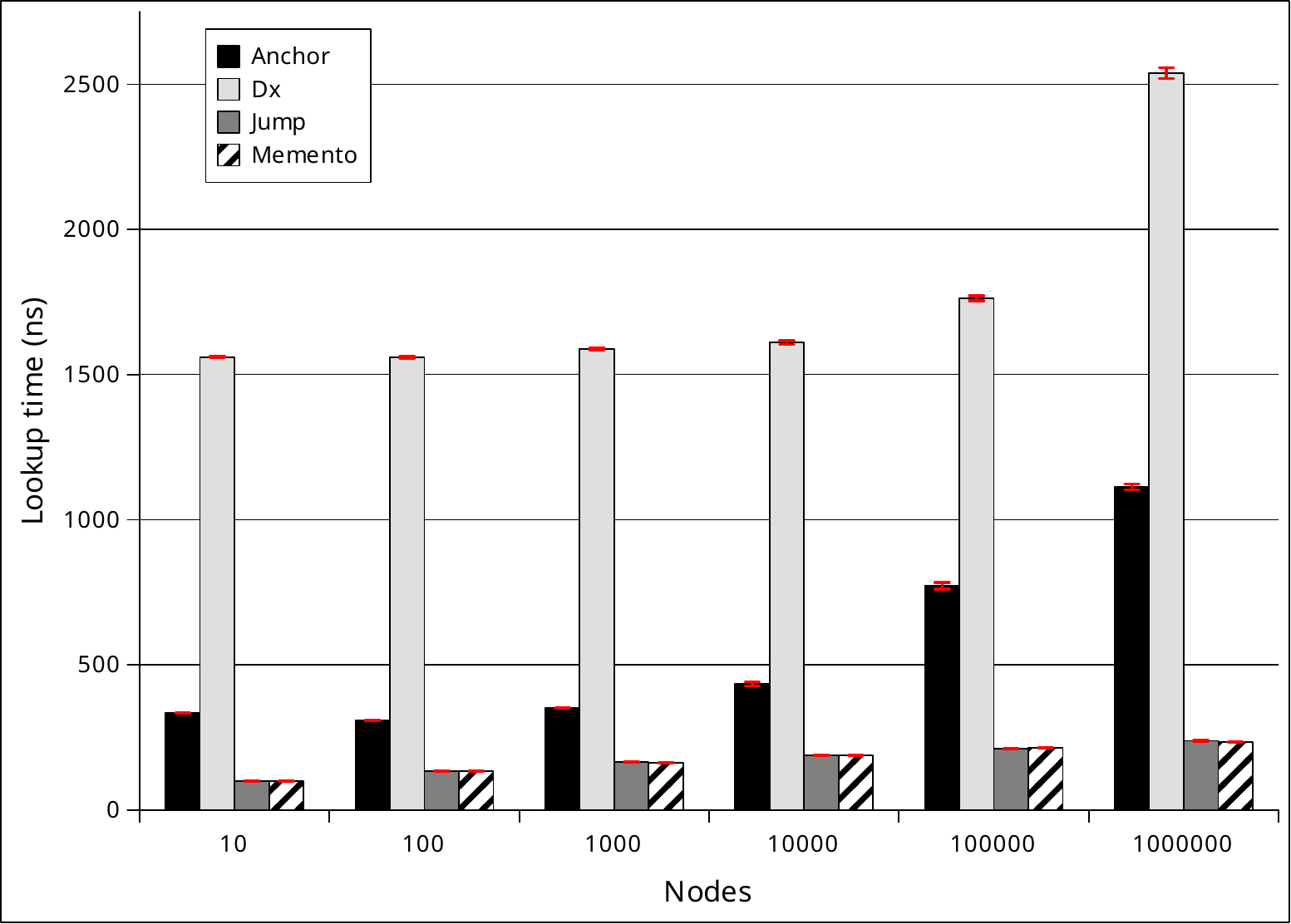}
    \caption{One-shot Removals - Lookup time (best case)}
    \label{fig:lookup_time_bc}
\end{figure}

\begin{figure}[ht!]
    \includegraphics[width=0.48\textwidth]{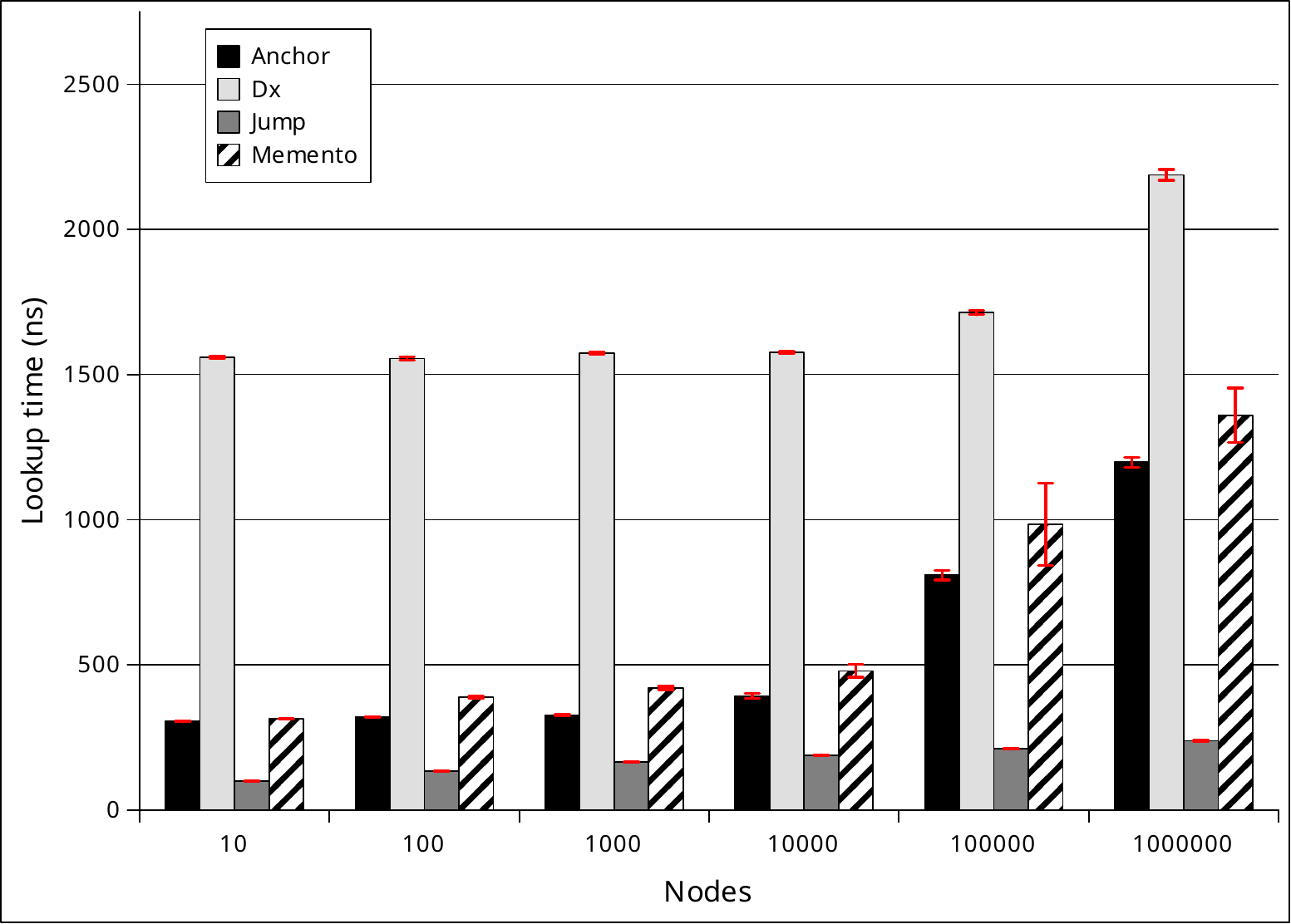}
    \caption{One-shot Removals - Lookup time (worst case)}
    \label{fig:lookup_time_wc}
\end{figure}

Concerning lookup-time, both \textit{Memento} and \textit{Jump} perform significantly better than \textit{Dx} and display better results also compared to \textit{Anchor} in the best-case scenario (Fig.~\ref{fig:lookup_time_bc}). The advantage of \textit{Memento} over \textit{Anchor} disappears in the worst-case scenario, where the latter algorithm performs slightly better on average (Fig.~\ref{fig:lookup_time_wc}). While \textit{Dx} remains the slowest of the pack. As mentioned before, \textit{Jump} does not support random removals. Therefore the reported measurement still concerns the \textit{LIFO} ordering.

\subsection{Incremental removals}
\label{sec:incremental}
In the last scenario, we simulate the incremental removal up to $90\%$ of the nodes. Fig. \ref{fig:lookup_time_incremental_bc} and Fig. \ref{fig:lookup_time_incremental_wc} report the lookup time in the best case and in the worst case, respectively. In the best case, and in contrast to the previous scenario, \textit{Dx} is by far the worst performer: this result can be explained by the fact that several additional resolution steps are introduced by that algorithm when removing nodes. \textit{Memento} and \textit{Jump} perform similarly, as the logic for both of them is similar. In the worst case, \textit{Anchor} is the worst performer, up until the removal of $65\%$ of the nodes. After that threshold, \textit{Memento} and \textit{Dx} become the slowest ones.
Losing more than $65\%$ of nodes is not a common situation. During the regular life cycle of a cluster, the number of failing nodes is likely never to exceed the threshold of $20\%$, which is the optimal performance range for \textit{Memento}.

\begin{figure}[ht!]
    \includegraphics[width=0.48\textwidth]{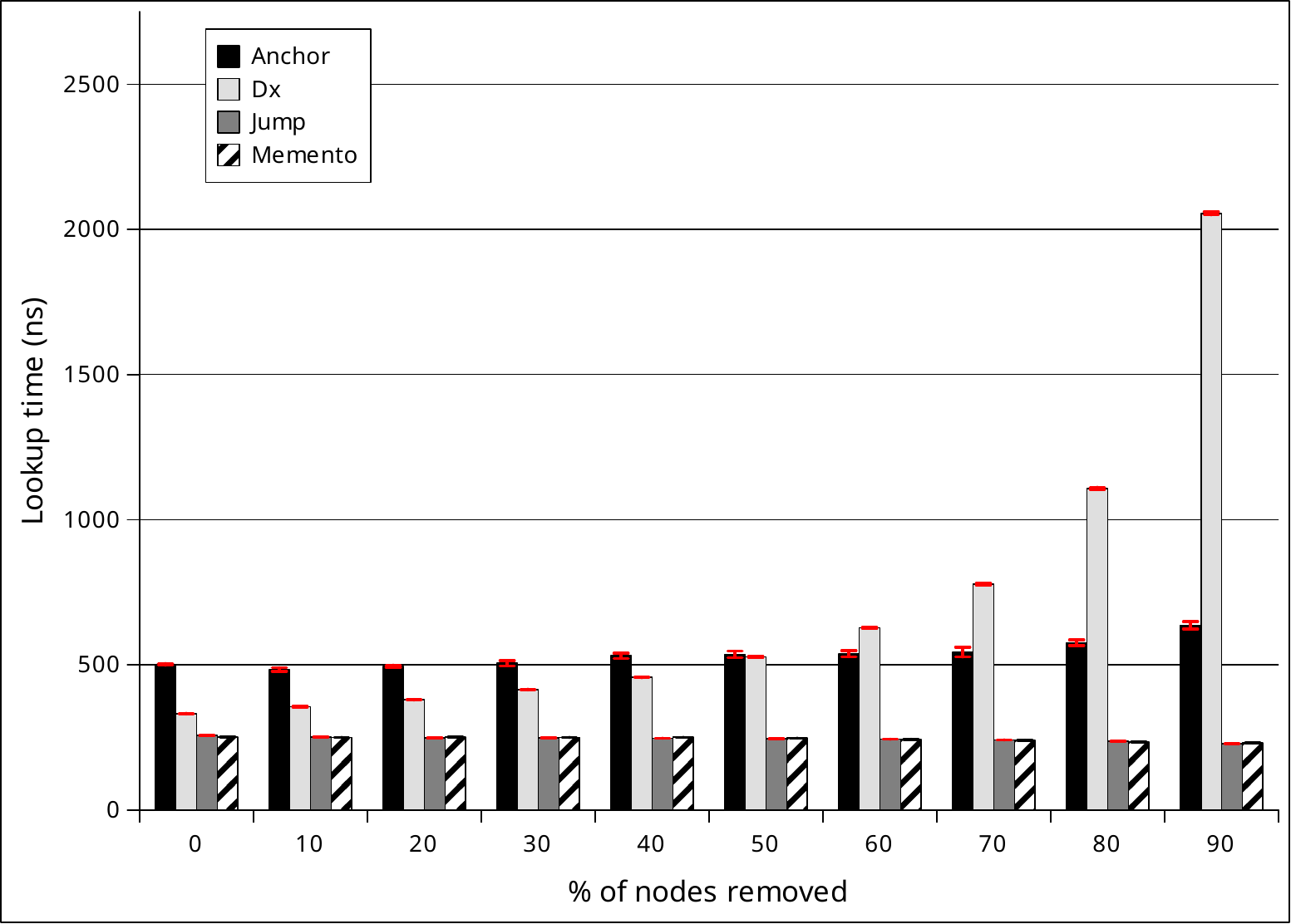}
    \caption{Incremental removals - Lookup time (best case)}
    \label{fig:lookup_time_incremental_bc}
\end{figure}

\begin{figure}[ht!]
    \includegraphics[width=0.48\textwidth]{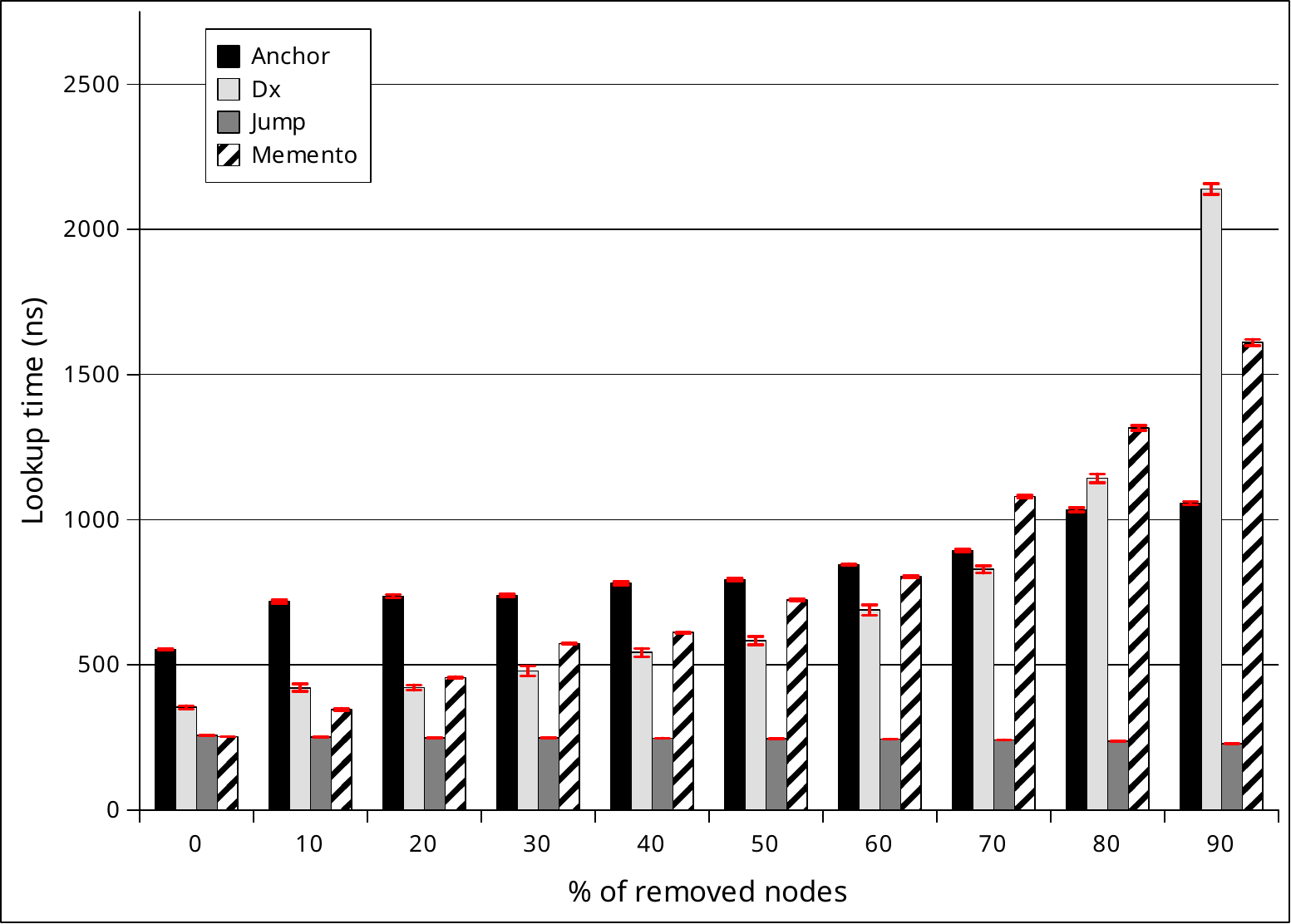}
    \caption{Incremental removals - Lookup time (worst case)}
    \label{fig:lookup_time_incremental_wc}
\end{figure}

Regarding memory usage, \textit{LIFO} removals do not affect neither \textit{Jump} nor \textit{Memento}, as no additional information is required to keep track of removed nodes (Fig. \ref{fig:memory_usage_incremental_bc}). In the worst case, analogously to the \textit{One-shot removals} scenario, \textit{Memento} exhibits similar results as \textit{Dx}, whereas \textit{Anchor} has the highest memory consumption.

\begin{figure}[ht!]
    \includegraphics[width=0.48\textwidth]{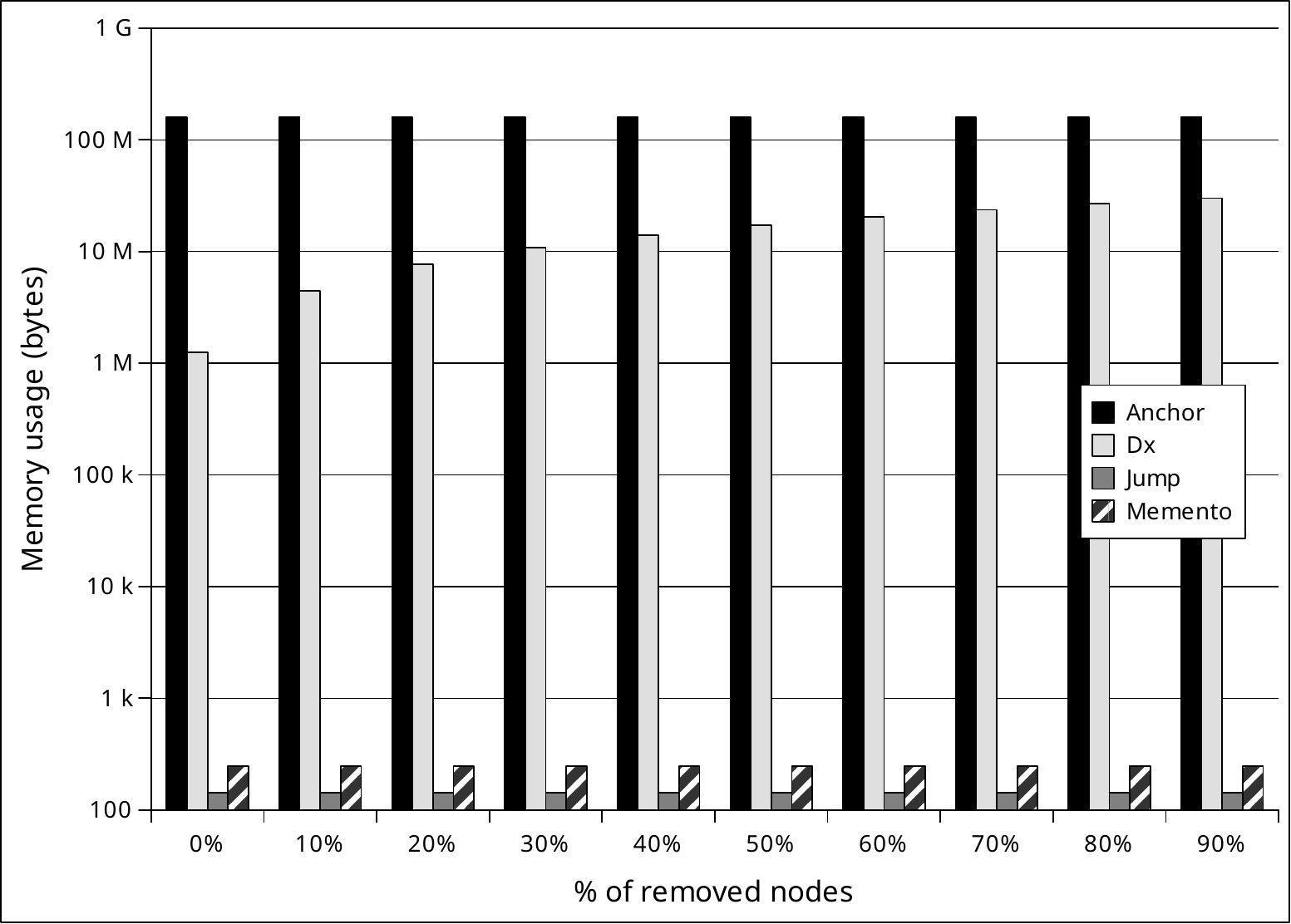}
    \caption{Incremental removals - Memory usage (best case)}
    \label{fig:memory_usage_incremental_bc}
\end{figure}

\begin{figure}[ht!]
    \includegraphics[width=0.48\textwidth]{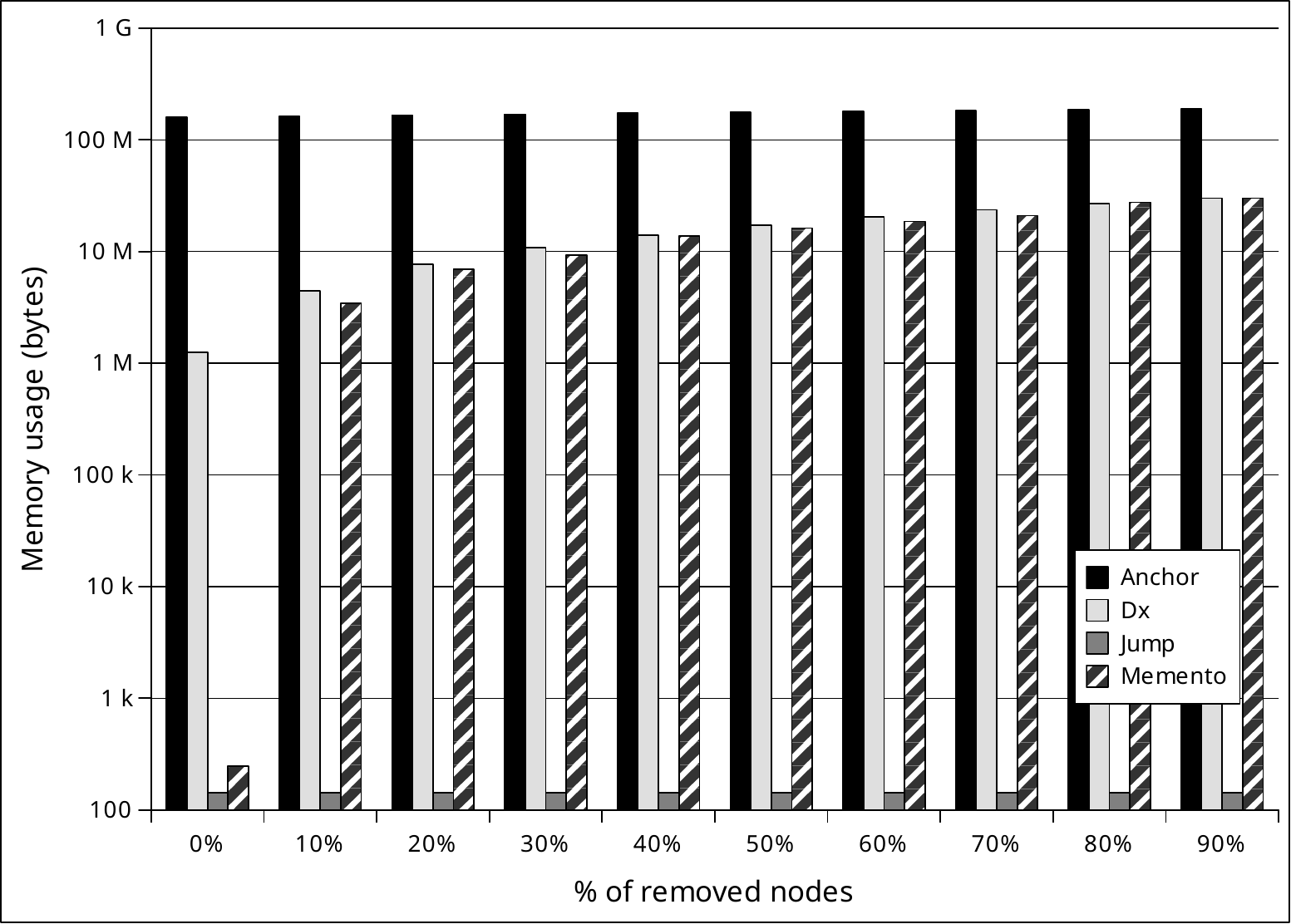}
    \caption{Incremental removals - Memory usage (worst case)}
    \label{fig:memory_usage_incremental_wc}
\end{figure}

\subsection{Sensitivity to the $\frac{a}{w}$ ratio in \textit{Anchor} and \textit{Dx}}
\label{sec:sensitivity} 
As stated at the beginning of Sec.~\ref{sec:benchmarks}, both \textit{Anchor} and \textit{Dx} require setting an upper bound for the cluster’s overall capacity. We argued that this value affects both time and memory complexity, and decided to perform our benchmarks using $\frac{a}{w} = 10$, which we consider a sensible option. To corroborate our choice, we present here a sensitivity analysis to determine how different ratios affect the performance of those algorithms.

A sensitivity analysis was conducted by keeping the number of working nodes $w$ fixed and by varying the overall capacity $a$. More specifically, we conducted experiments on the following scenarios:
\begin{itemize}
\item \textit{Stable}: we consider a cluster of one million working nodes ($w=10^6$);
\item \textit{One-shot removal of 20\% of the nodes}: from an initial size of one million working nodes ($w=10^6$), we randomly remove $20\%$ of the nodes (effectively resulting in a scenario where only $800\ 000$ nodes are operative during lookup);
\item \textit{One-shot removal of 65\% of the nodes}: from an initial size of one million working nodes ($w=10^6$), we randomly remove $65\%$ of the nodes (effectively resulting in a scenario where only $350\ 000$ nodes are operative during lookup).
\end{itemize} 

The considered $\frac{a}{w}$ ratios are $5$, $10$, $20$, $50$, and $100$ (where $w$ is always the initial size of $10^6$ nodes). Both lookup time and memory usage where measured.
Beside the \textit{stable} scenario, we choose to study one-shot removals of $20\%$ and $65\%$ of the nodes. The former can be considered as a realistic worst-case threshold for the number of failing nodes during the regular life cycle of a cluster, whereas the latter was chosen because it exhibits similar lookup performance between \textit{Memento}, \textit{Anchor}, and \textit{Dx} (as shown in Sec.~\ref{sec:incremental}). The results also report the values for \textit{Memento} as a baseline.

\begin{figure}[ht!]
    \includegraphics[width=0.48\textwidth]{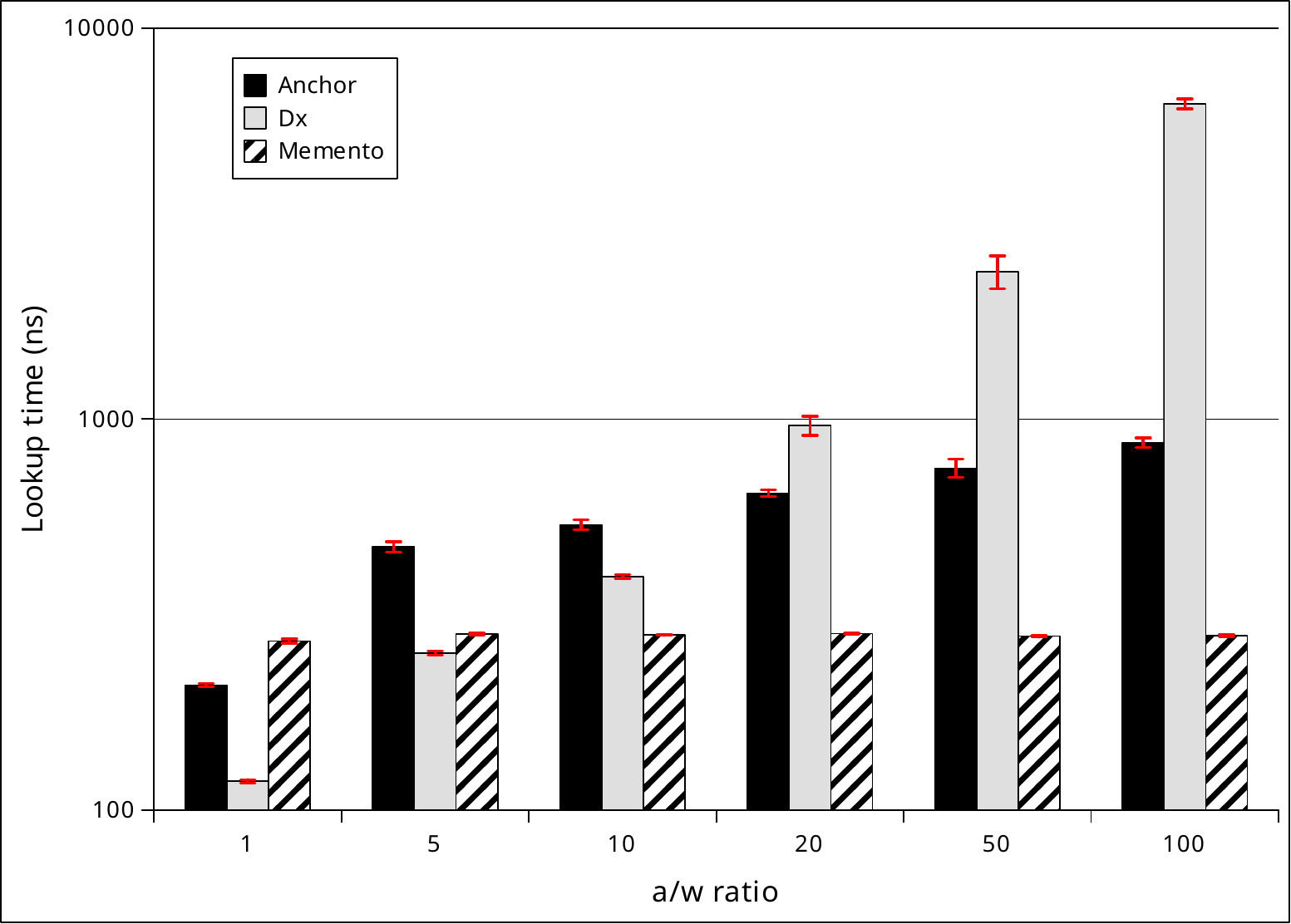}
    \caption{Sensitivity to $\frac{a}{w}$ - Lookup time (\textit{Stable})}
    \label{fig:sensitivity_lookuptime_stable}
\end{figure}

\begin{figure}[ht!]
    \includegraphics[width=0.48\textwidth]{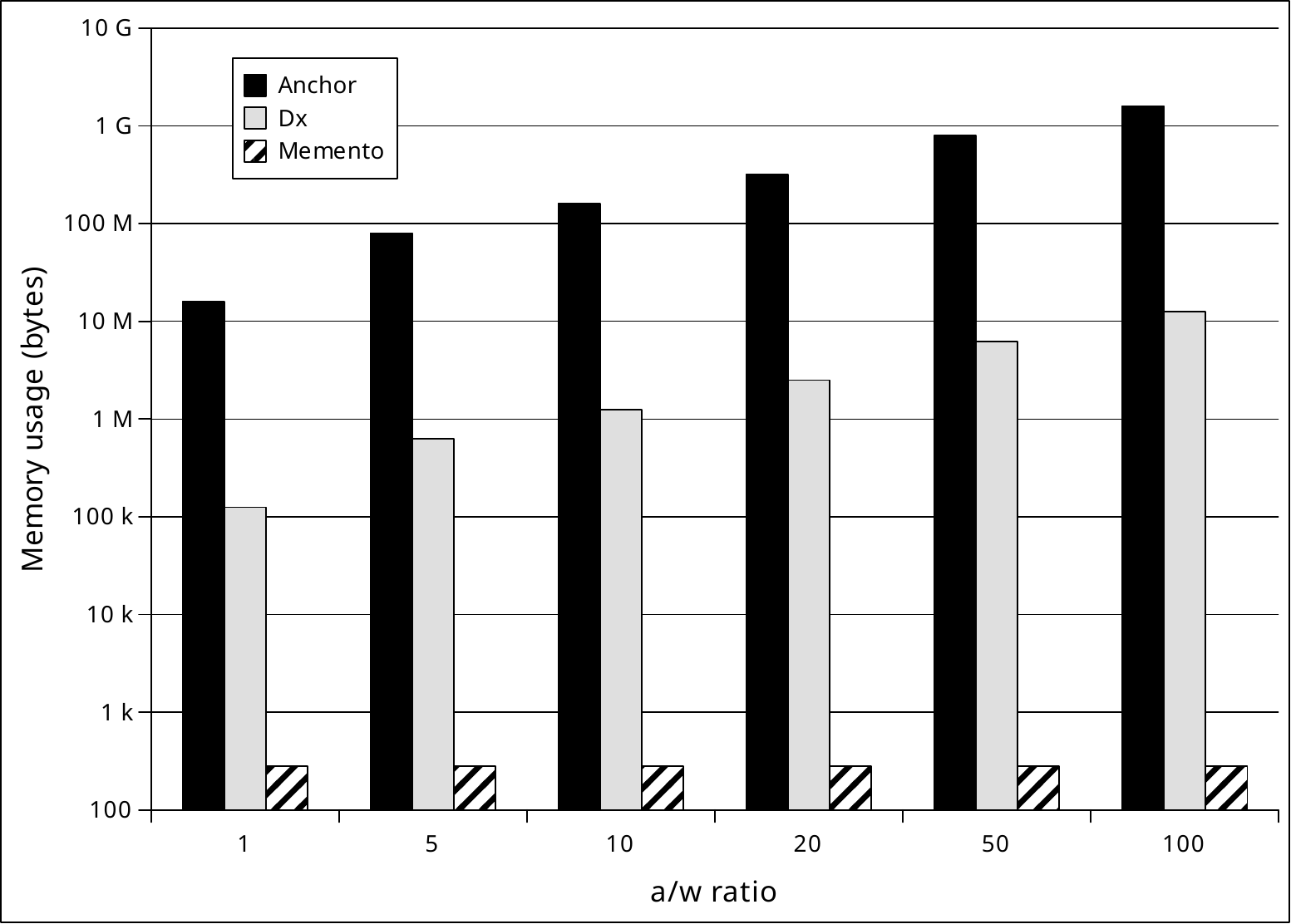}
    \caption{Sensitivity to $\frac{a}{w}$ - Memory usage (\textit{Stable})}
    \label{fig:sensitivity_memory_stable}
\end{figure}

In the \textit{Stable} scenario, the performance of \textit{Dx} is heavily affected by large over-provisioning ratios, with the lookup time growing linearly (as shown in Fig.~\ref{fig:sensitivity_lookuptime_stable}). On the contrary, \textit{Anchor} continues to perform well, although showing a logarithmic growth as the ratio increases (as expected from the asymptotic complexity). With respect to memory usage, both \textit{Anchor} and \textit{Dx} suffer considerably from large ratios, as shown in Fig.~\ref{fig:sensitivity_memory_stable}. In practice, the need for preallocating data structures based on the expected overall capacity, and the penalties incurred when overestimating the maximum scale of the system could represent a significant hindrance when deploying solutions based on \textit{Anchor} or \textit{Dx}. On the contrary, \textit{Memento} does not limit the overall capacity of the cluster and always maintains minimal memory usage and optimal time performance.

\begin{figure}[ht!]
    \includegraphics[width=0.48\textwidth]{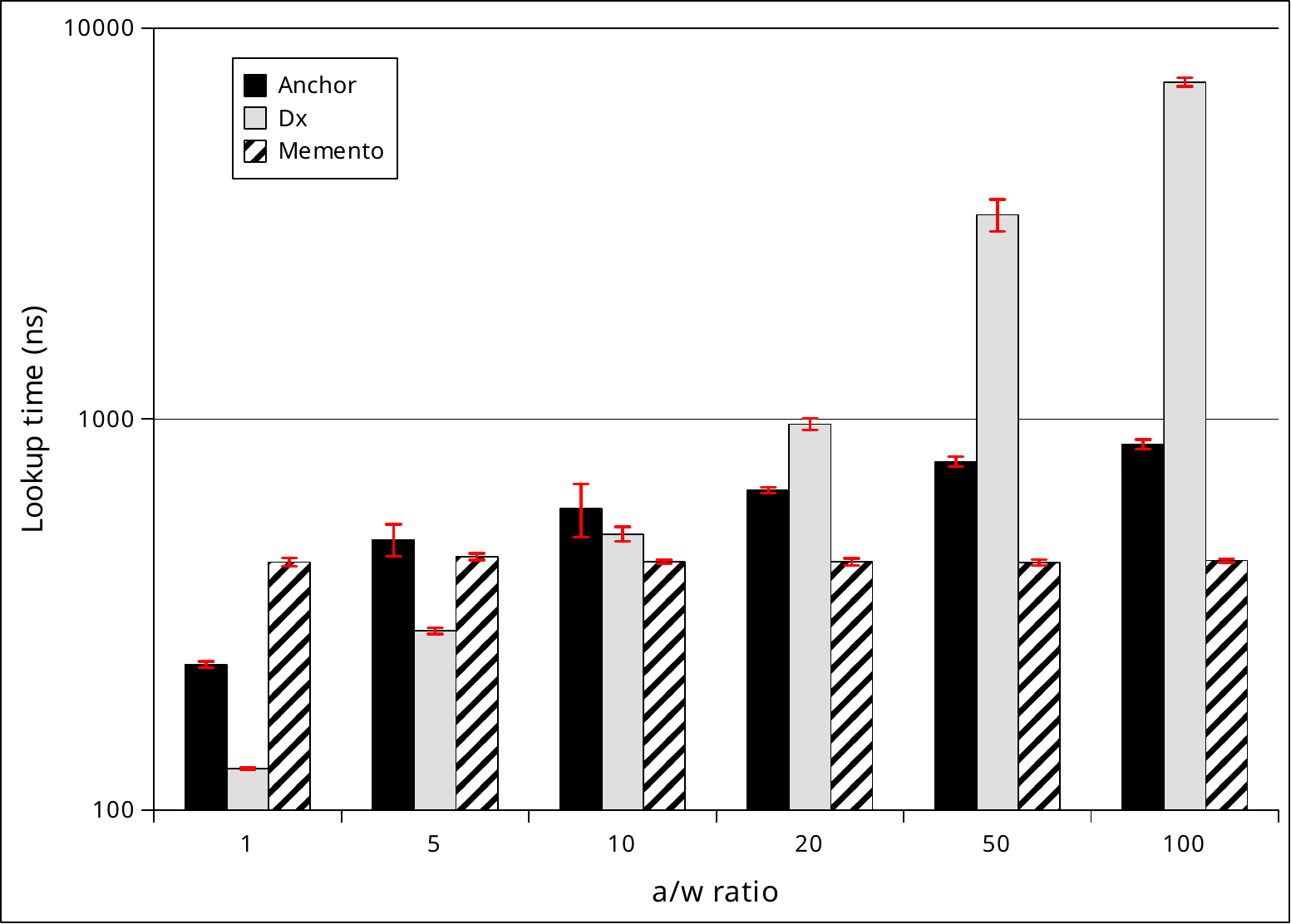}
    \caption{Sensitivity to $\frac{a}{w}$ - Lookup time (\textit{20\% removals})}
    \label{fig:sensitivity_lookuptime_20}
\end{figure}

\begin{figure}[ht!]
    \includegraphics[width=0.48\textwidth]{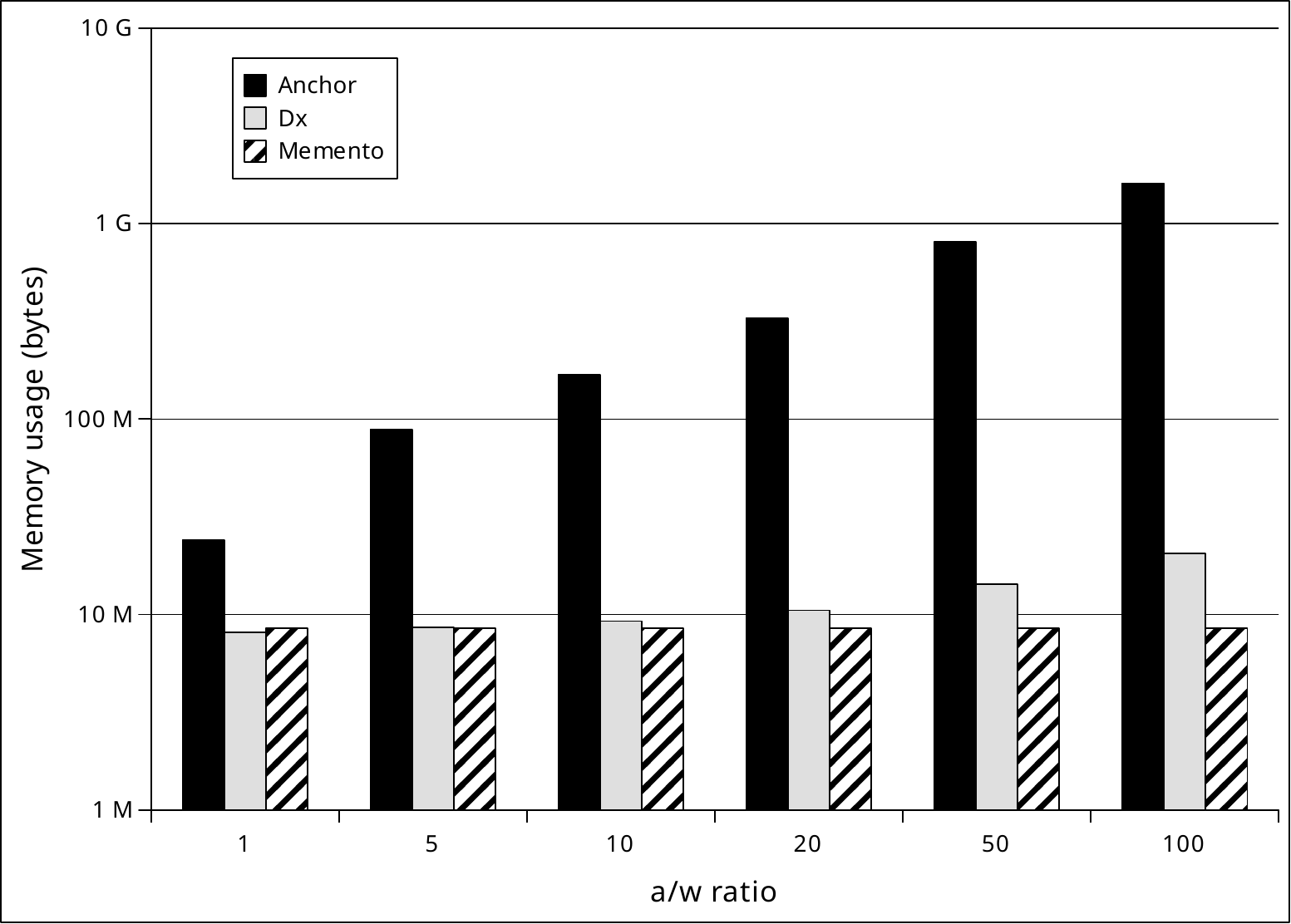}
    \caption{Sensitivity to $\frac{a}{w}$ - Memory usage (\textit{20\% removals})}
    \label{fig:sensitivity_memory_20}
\end{figure}

As expected, the performance of all the considered algorithms worsens as nodes are removed from the system. In both evaluation scenarios with one-shot removal of a percentage of the nodes, it is possible to observe a linear growth of the lookup time for \textit{Dx}, and a logarithmic growth for \textit{Anchor} as the ratio $\frac{a}{w}$ increases (Figs.~\ref{fig:sensitivity_lookuptime_20} and \ref{fig:sensitivity_lookuptime_65}). It is also worth noting that with $\frac{a}{w} = 10$ and $65\%$ of the nodes removed, the lookup performance of all three considered algorithms is almost equal. Meanwhile, with only $20\%$ of the nodes removed \textit{Anchors} is the worst performer.

\begin{figure}[ht!]
    \includegraphics[width=0.48\textwidth]{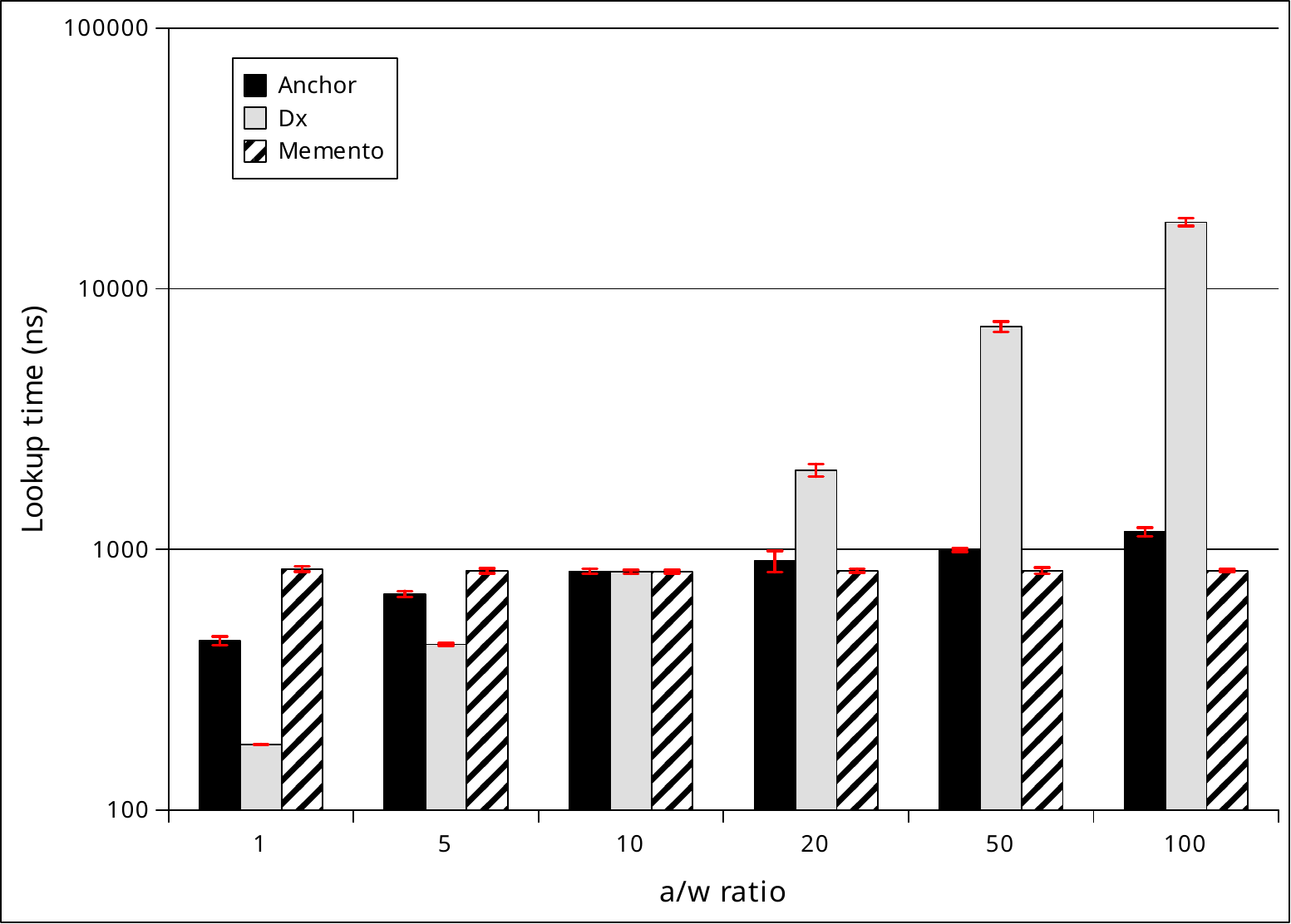}
    \caption{Sensitivity to $\frac{a}{w}$ - Lookup time (\textit{65\% removals})}
    \label{fig:sensitivity_lookuptime_65}
\end{figure}

\begin{figure}[ht!]
    \includegraphics[width=0.48\textwidth]{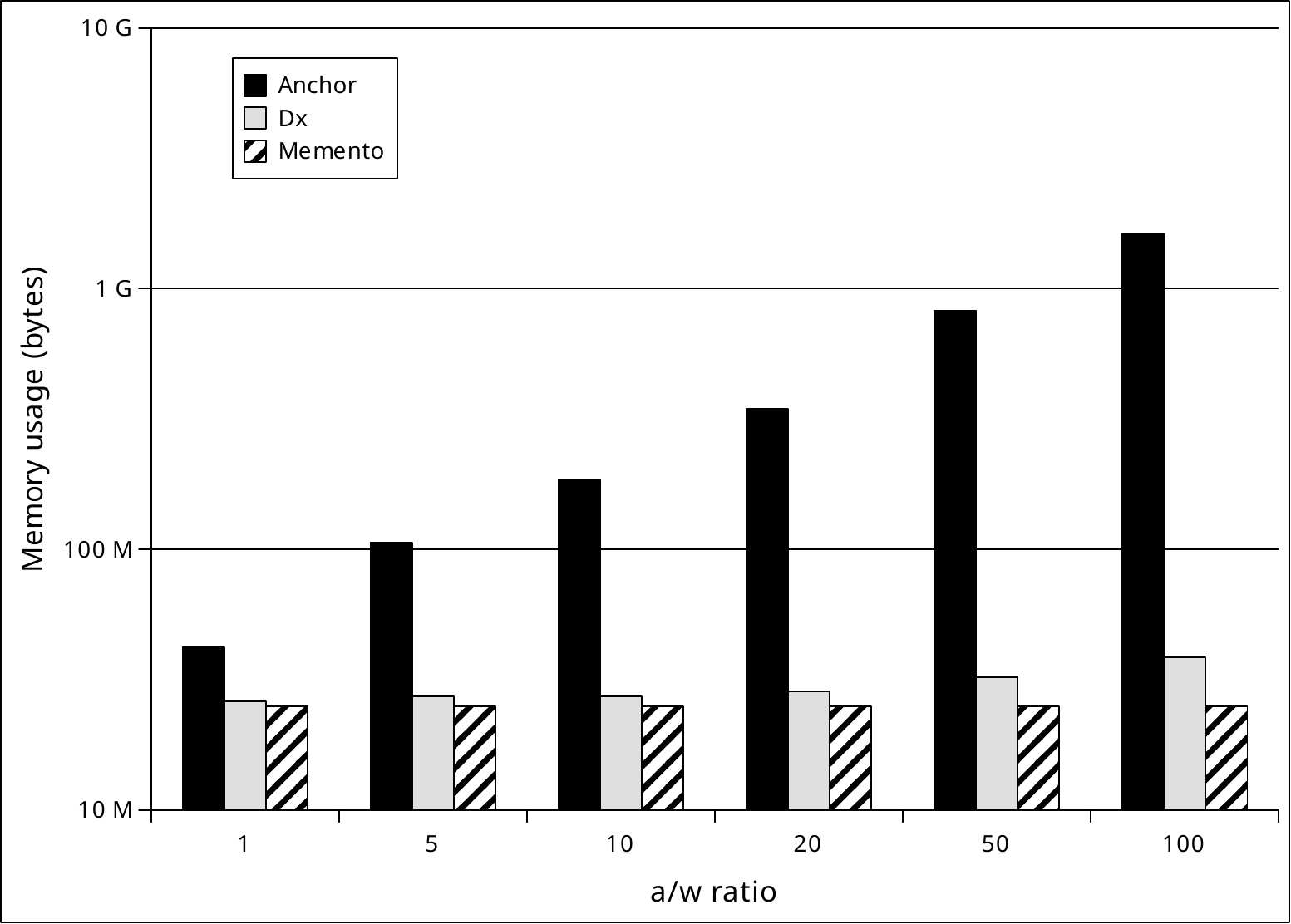}
    \caption{Sensitivity to $\frac{a}{w}$ - Memory usage (\textit{65\% removals})}
    \label{fig:sensitivity_memory_65}
\end{figure}

Concerning memory usage, it comes at no surprise that both \textit{Anchor} and \textit{Dx} show similar results in all scenarios (Figs.~\ref{fig:sensitivity_memory_20} and \ref{fig:sensitivity_memory_65}). The slight differences which can be observed as the number of removals increases (from no removals, up to 65\%) can be explained by the need for storing the order of the removals.

\subsection{Discussion}
\textit{Memento} shows to be asymptotically slower than the other considered algorithms. However, as shown in this section, it is faster in practice. The asymptotic complexity is given by two parts:
\begin{enumerate}
    \item $O(ln(n))$ which is given by the initial execution of \textit{Jump} (Alg.~\ref{alg:lookup} line $2$), and
    \item $O([ln(\frac{n}{w})]^2)$ which is given by the two nested loops (Alg.~\ref{alg:lookup} lines $3$ to $11$).
\end{enumerate} 
As discussed in Sec.~\ref{sec:jump}, \textit{Jump} achieves high speed by avoiding memory accesses and operating at CPU speed. This is evident from the comparison shown in Figs.~\ref{fig:lookup_time_incremental_bc} and \ref{fig:lookup_time_incremental_wc}, where \textit{Jump} consistently outperforms all other algorithms. The \textit{Memento} algorithm incorporates \textit{Jump} as its initial instruction and enters the nested loops only if the bucket returned by \textit{Jump} is not working. When all buckets are working, \textit{Memento} performs similarly to \textit{Jump}. However, as buckets are removed, the latter part of the complexity becomes increasingly significant. Nonetheless, \textit{Memento} exhibits better performance than \textit{Anchor} and \textit{Dx} until approximately $65\%$ of the buckets are removed. In real-world production environments, significant efforts are made to prevent failures, resulting in a low probability of concurrent failures. Additionally, if more than $20\%$ of the nodes fail simultaneously, the performance or memory consumption of the consistent hashing algorithm becomes a secondary concern. On the whole, we argue that \textit{Memento} has an edge over \textit{Anchor} and \textit{Dx} by not requiring an upper limit on the cluster size: as shown in our sensitivity analysis, it is evident how larger upper bounds have a significant impact on the memory requirements of \textit{Anchor} and on the lookup performance of \textit{Dx}.
The recommended usage pattern for \textit{Memento} involves scaling the cluster by adding and removing buckets in a \textit{LIFO} order, utilizing replacements exclusively for failures. This approach ensures that the internal structure remains empty, resulting in optimal performance and minimal memory usage. 

\section{Conclusion}

In this paper, we presented a novel consistent hashing algorithm named \textit{MementoHash}, which improves upon state-of-the-art algorithms such as \textit{JumpHash}, \textit{AnchorHash}, and \textit{DxHash}.
We provided implementation details and theoretical guarantees. \textit{MementoHash} is an advancement of \textit{JumpHash}, capable of handling random failures while maintaining minimal memory usage and optimal time performance. Unlike \textit{AnchorHash} and \textit{DxHash}, it does not require fixing the overall capacity of the cluster, thus allowing for indefinite scalability. The recommended usage pattern for \textit{Memento} involves scaling the cluster by adding and removing buckets in a Last-In-First-Out order, utilizing replacements exclusively for failures. This approach ensures that the internal structure remains empty, resulting in optimal performance and minimal memory usage. By means of an experimental evaluation, we measured and compared the behavior of those algorithms concerning lookup time and memory usage. Results show that the lookup time of \textit{MementoHash} is on par with \textit{JumpHash} and considerably better than both \textit{AnchorHash} and \textit{DxHash}. In terms of memory usage, \textit{MementoHash} ranks second only to \textit{JumpHash}. However, \textit{JumpHash} does not allow random removals, and only the last added node can be removed. To address this limitation, \textit{AnchorHash} and \textit{DxHash} keep track of all the nodes of the cluster (working and not working). To ensure balance and minimal disruption, those algorithms force an upper bound to the overall capacity of the cluster, leading to reduced scalability.
\textit{MementoHash} overcomes the limitations of \textit{JumpHash} without limiting the overall capacity of the cluster. Compared to \textit{AnchorHash} and \textit{DxHash}, our approach exhibits similar memory usage in the worst-case scenario and performs significantly better in the best-case scenario.

\section{Future work}
In order to further expand the scope of application of \textit{MementoHash} and facilitate its implementation in existing or novel distributed solutions, several avenues of future work can be explored. First and foremost, it is worth investigating the feasibility of extending \textit{MementoHash} to maintain full consistency when nodes may not unanimously agree on the removal order. Furthermore, we aim at investigating the applicability of our solution to a scenario with bounded loads \cite{https://doi.org/10.48550/arxiv.1608.01350}.

\bibliographystyle{IEEEtran}
\bibliography{bibliography}

\begin{thebibliography}{10}
\providecommand{\url}[1]{#1}
\csname url@samestyle\endcsname
\providecommand{\newblock}{\relax}
\providecommand{\bibinfo}[2]{#2}
\providecommand{\BIBentrySTDinterwordspacing}{\spaceskip=0pt\relax}
\providecommand{\BIBentryALTinterwordstretchfactor}{4}
\providecommand{\BIBentryALTinterwordspacing}{\spaceskip=\fontdimen2\font plus
\BIBentryALTinterwordstretchfactor\fontdimen3\font minus
  \fontdimen4\font\relax}
\providecommand{\BIBforeignlanguage}[2]{{%
\expandafter\ifx\csname l@#1\endcsname\relax
\typeout{** WARNING: IEEEtran.bst: No hyphenation pattern has been}%
\typeout{** loaded for the language `#1'. Using the pattern for}%
\typeout{** the default language instead.}%
\else
\language=\csname l@#1\endcsname
\fi
#2}}
\providecommand{\BIBdecl}{\relax}
\BIBdecl

\bibitem{saxenaanalysis}
S.~Priyamvada, ``Analysis of various hash function,'' \emph{International
  Journal of Innovative Science and Research Technology}, vol.~3, 2018.

\bibitem{thaler1996rendezvous}
D.~Thaler and C.~V. Ravishankar, ``A name-based mapping scheme for
  rendezvous,'' in \emph{Technical Report CSE-TR-316-96, University of
  Michigan}.\hskip 1em plus 0.5em minus 0.4em\relax University of Michigan,
  1996.

\bibitem{thaler1998using}
D.~G. Thaler and C.~V. Ravishankar, ``Using name-based mappings to increase hit
  rates,'' \emph{IEEE/ACM Transactions on networking}, vol.~6, no.~1, pp.
  1--14, 1998.

\bibitem{karger1997consistent}
D.~Karger, E.~Lehman, T.~Leighton, R.~Panigrahy, M.~Levine, and D.~Lewin,
  ``Consistent hashing and random trees: Distributed caching protocols for
  relieving hot spots on the world wide web,'' in \emph{Proceedings of the
  twenty-ninth annual ACM symposium on Theory of computing}, 1997, pp.
  654--663.

\bibitem{karger1999consistent}
D.~Karger, A.~Sherman, A.~Berkheimer, B.~Bogstad, R.~Dhanidina, K.~Iwamoto,
  B.~Kim, L.~Matkins, and Y.~Yerushalmi, ``Web caching with consistent
  hashing,'' \emph{Comput. Netw.}, vol.~31, no. 11–16, p. 1203–1213, may
  1999.

\bibitem{lamping2014fast}
J.~Lamping and E.~Veach, ``A fast, minimal memory, consistent hash algorithm,''
  \emph{arXiv preprint arXiv:1406.2294}, 2014.

\bibitem{appleton2015multi}
B.~Appleton and M.~O'Reilly, ``Multi-probe consistent hashing,'' \emph{arXiv
  preprint arXiv:1505.00062}, 2015.

\bibitem{eisenbud2016maglev}
D.~E. Eisenbud, C.~Yi, C.~Contavalli, C.~Smith, R.~Kononov, E.~Mann-Hielscher,
  A.~Cilingiroglu, B.~Cheyney, W.~Shang, and J.~D. Hosein, ``Maglev: A fast and
  reliable software network load balancer,'' in \emph{13th $\{$USENIX$\}$
  Symposium on Networked Systems Design and Implementation ($\{$NSDI$\}$ 16)},
  2016, pp. 523--535.

\bibitem{mendelson2020anchorhash}
G.~Mendelson, S.~Vargaftik, K.~Barabash, D.~H. Lorenz, I.~Keslassy, and
  A.~Orda, ``Anchorhash: A scalable consistent hash,'' \emph{IEEE/ACM
  Transactions on Networking}, vol.~29, no.~2, pp. 517--528, 2020.

\bibitem{dong2021dxhash}
C.~Dong and F.~Wang, ``Dxhash: A scalable consistent hash based on the
  pseudo-random sequence,'' \emph{arXiv preprint arXiv:2107.07930}, 2021.

\bibitem{posterIspass2023}
M.~Coluzzi, A.~Brocco, P.~Contu, and T.~Leidi, ``A survey and comparison of
  consistent hashing algorithms,'' in \emph{2023 IEEE International Symposium
  on Performance Analysis of Systems and Software}, April 2023.

\bibitem{sebd2023}
M.~Coluzzi, A.~Brocco, and T.~Leidi, ``Consistently faster: A survey and fair
  comparison of consistent hashing algorithms,'' in \emph{Proceedings of SEBD
  2023: 31st Symposium on Advanced Database System, July 02--05, 2023,
  Galzignano Terme, Padua, Italy}, ser. {CEUR} Workshop Proceedings, 2023.

\bibitem{isinGitHub}
\BIBentryALTinterwordspacing
{Institute of Information Systems and Networking at SUPSI}.
  java-consistent-hashing-algorithms. [Online]. Available:
  \url{https://github.com/SUPSI-DTI-ISIN/java-consistent-hashing-algorithms}
\BIBentrySTDinterwordspacing

\bibitem{cressie1986moment}
N.~Cressie and M.~Borkent, ``The moment generating function has its moments,''
  \emph{Journal of Statistical Planning and Inference}, vol.~13, pp. 337--344,
  1986.

\bibitem{sun2012arithmetic}
Z.-W. Sun, ``Arithmetic theory of harmonic numbers,'' \emph{Proceedings of the
  American Mathematical Society}, vol. 140, no.~2, pp. 415--428, 2012.

\bibitem{https://doi.org/10.48550/arxiv.1608.01350}
\BIBentryALTinterwordspacing
V.~Mirrokni, M.~Thorup, and M.~Zadimoghaddam, ``Consistent hashing with bounded
  loads,'' 2016. [Online]. Available: \url{https://arxiv.org/abs/1608.01350}
\BIBentrySTDinterwordspacing

\end{thebibliography}
\balance
\end{document}